\documentclass[12pt]{article}
\usepackage{amsmath}
\usepackage{amssymb}
\usepackage{graphicx}
\usepackage{xcolor}
\usepackage{subcaption} 
\usepackage{natbib}
\usepackage{enumitem}
\usepackage{authblk}
\usepackage{url} 
\usepackage[pdfencoding=auto, psdextra]{hyperref}
\hypersetup{
    colorlinks,
    linkcolor={blue!70!black},
    citecolor={blue!70!black},
    urlcolor={blue!70!black}
}

\usepackage{amsthm} 
\newtheorem{theorem}{Theorem}[section]

\newtheorem{proposition}[theorem]{Proposition}

\usepackage{mathrsfs}
\usepackage{placeins} 
\usepackage{booktabs} 
\usepackage{algorithm} 
\usepackage{algorithmic}

\newcommand{\blind}{0}

\newcommand{\kld}{\mathop{\text{\footnotesize KLD}}}
\newcommand{\iae}{\mathop{\text{\footnotesize IAE}}}
\newcommand{\auc}{\mathop{\text{\footnotesize AUC}}}

\newcommand{\bftz}{\mathbf{\tilde{z}}}
\newcommand{\bfth}{\mathbf{\tilde{h}}}

\newcommand{\tz}{\tilde{z}}

\newcommand{\thh}{\tilde{h}}

\newcommand{\independence}{\perp\!\!\!\!\perp}  

\newcommand{\bsh}{\boldsymbol{h}}

\newcommand{\bfh}{\mathbf{h}}

\newcommand{\bfx}{\mathbf{x}}
\newcommand{\bfy}{\mathbf{y}}
\newcommand{\bfz}{\mathbf{z}}

\newcommand{\bfI}{\mathbf{I}}

\newcommand{\bfP}{\mathbf{P}}

\newcommand{\bfY}{\mathbf{Y}}

\newcommand{\sfT}{\mathsf{T}}

\newcommand{\bsupsilon}{\boldsymbol{\upsilon}}

\newcommand{\bstheta}{\boldsymbol{\theta}}

\newcommand{\bsLambda}{\boldsymbol{\Lambda}}

\newcommand{\bsSigma}{\boldsymbol{\Sigma}}

\newcommand{\bfzero}{\mathbf{0}}

\newcommand{\calN}{\mathcal{N}}

\newcommand{\bbR}{\mathbb{R}}

\begin{document}

\def\spacingset#1{\renewcommand{\baselinestretch}
{#1}\small\normalsize} \spacingset{1}

\if0\blind
{
  \title{\bf Graphical Transformation Models}
  \author[1,2,3]{Matthias Herp\thanks{
  The authors MA and MH were supported by the German Federal Ministry of Education and Research (BMBF) [01KD2209D]. \\
  The authors also thank Thomas Nagler, from the LMU in Munich, for fruitful discussions and feedback on identifying conditional independencies in vine copulas. \\
  Address questions to \texttt{matthias.herp@bioinf.med.uni-goettingen.de}}
}
\author[1]{Johannes Brachem}
\author[3]{Michael Altenbuchinger}
\author[1,2]{Thomas Kneib}

\renewcommand*{\Affilfont}{\small\normalfont}
\affil[1]{Chair of Statistics, University of Göttingen, Germany}
\affil[2]{Campus Institute Data Science,
University of Göttingen, Germany}
\affil[3]{Department of Medical Bioinformatics, University Medical Center Göttingen, Germany}
  \maketitle
} \fi

\bigskip
\begin{abstract}
Graphical Transformation Models (GTMs) are introduced as a novel approach to effectively model multivariate data with intricate marginals and complex dependency structures semiparametrically, while maintaining interpretability through the identification of varying conditional independencies. GTMs extend multivariate transformation models by replacing the Gaussian copula with a custom-designed multivariate transformation, offering two major advantages. Firstly, GTMs can capture more complex interdependencies using penalized splines, which also provide an efficient regularization scheme. Secondly, we demonstrate how to approximately regularize GTMs towards pairwise conditional independencies using a lasso penalty, akin to Gaussian graphical models. The model's robustness and effectiveness are validated through simulations, showcasing its ability to accurately learn complex dependencies and identify conditional independencies. Additionally, the model is applied to a benchmark astrophysics dataset, where the GTM demonstrates favorable performance compared to non-parametric vine copulas in learning complex multivariate distributions.
\end{abstract}


\vfill

\section{Introduction}\label{sec:introduction}

Multivariate models allow researchers to examine patterns in the joint behavior of multiple variables, particularly in the diverse -omics fields dealing with high-dimensional biological data, such as genomics and microbiomics. A common challenge with multivariate models is navigating the trade-offs between restricted models that offer interpretability advantages, such as Gaussian Graphical Models (GGMs) \citep{ggm_1_lauritzen, ggm_2_bishop}, and highly flexible models capable of capturing complex relationships, such as Normalizing Flows \citep{DeepMind_NF_Review, NF_Review_Kobyzev}.

GGMs formulate the problem as estimating a multivariate Gaussian distribution. A valuable feature of these models is that zero entries in the precision matrix of a Gaussian distribution can be interpreted as conditional independencies between the two corresponding variables given all other variables. The precision matrix entries can then be penalized, for instance, with the least absolute shrinkage operator (LASSO) to encourage sparse precision matrices \citep{Graph_Lasso, graph_lasso_Tibshirani, lasso}. This is crucial because, in many high-dimensional contexts, it is reasonable that any single variable can be predicted effectively based on a small number of other variables—gene expression analysis serves as a prime example \citep{Dobra2004-SparseGraphicalModels}. Other common applications of GGMs include network analysis in neuroimaging \citep[e.g.,][]{Rosa2015-SparseNetworkbasedModels} and exploring microbial interactions \citep[e.g.,][]{Kurtz2015-SparseCompositionallyRobust}.

Although the capability to uncover conditional independencies is key to the popularity of GGMs, it comes at the cost of limited flexibility: in a multivariate Gaussian model, all variable relationships are inherently assumed to be linear. This restriction hampers the model's ability to capture significant real-world patterns, which are frequently nonlinear. For example, \cite{Souto-Maior2023-NonlinearExpressionPatterns} demonstrated that nonlinear gene expression patterns and multiple shifts in gene network interactions underlie major changes in sleep duration.

Normalizing Flows are a deep learning technique firmly established in the machine learning community and represent the opposite end of the modeling spectrum, adopting a maximalist approach to flexibility. In a normalizing flow, a multivariate distribution is represented through a sequence—or \textit{flow}—of invertible transformations that jointly map the data into a reference space, often a standard Gaussian. With an appropriate configuration of layers, normalizing flows exhibit remarkable capabilities in representing complex, highly non-Gaussian multivariate distributions.
Normalizing flows have been applied to various tasks, including image generation \citep{Kingma2018-GlowGenerativeFlow}, molecular graph generation in drug discovery \citep{Shi2020-GraphAFFlowbasedAutoregressive}, and event simulation in particle physics \citep{Gao2020-EventGenerationNormalizing}. Although normalizing flows are excellent for generative tasks, their inherent flexibility results in a lack of interpretability compared to GGMs. Moreover, they often require large amounts of training data.

Between these two poles lies a rich body of literature on copula models, which enable researchers to specify separate marginal distributions for each dimension of the multivariate target variable. These collections of univariate margins are then used to transform all dimensions to a common scale, typically standard uniform, via the probability integral transformation. The copula subsequently characterizes the relationships among the standardized margins. A Gaussian copula with Gaussian marginal distributions is equivalent to a Gaussian graphical model, and by allowing the selection of different margins and/or copula functions, researchers achieve additional flexibility. Simple copula models are generally suitable for low-dimensional data but can be expanded to accommodate higher dimensions through vine copulas \citep[see][for a review]{review_vine_czado_nagler}. Diagnosing and penalizing conditional independence in vine copulas is a complex task. \cite{Vine_Lasso} have proposed a sparsity-inducing approach for vine copulas based on relations to structural equation models using LASSO penalties; however, this method still selects dependencies based on linearity assumptions.

In copula models, commonly parametric margins can be replaced with semiparametric margins to create Multivariate Conditional Transformation Models \citep[MCTMs;][]{MCTM}. In MCTMs, the univariate marginal distributions are estimated directly from the observed data through independent, invertible transformations to a reference distribution \citep[see also][]{CTM, MLT}. This approach renders the marginal distributions both flexible and easy to implement, as they do not need to be manually selected. The transformations and dependence parameters can be conditioned on covariates, which is why the model is termed \textit{conditional}. When combined with a Gaussian copula, as proposed by \cite{MCTM}, an MCTM can be interpreted as a Gaussian graphical model at the level of the transformed marginals, allowing for the straightforward discovery of conditional independencies, similar to the copula graphical models introduced by \cite{Dobra2011} and the nonparametric graphs proposed by \cite{wassermann_nonpar}. Nonetheless, although such an MCTM offers added flexibility in the margins, the dependence structure among the transformed margins remains constrained by linearity assumptions.

In this paper, we leverage the concept of a sequence of transformations to create a fully semiparametric graphical multivariate transformation model. This model facilitates complex dependencies and the discovery of conditional independence, with minimal overhead in model specification. Specifically, we offer the following contributions:
\begin{enumerate}
    \item We introduce a semiparametric dependence model as a replacement for the Gaussian copula in an MCTM with flexible margins. This model is inspired by the sequential transformation concept utilized in normalizing flows.
    \item We present two targeted penalization schemes that are applied concurrently. The first scheme provides regularization towards conditional independence using the group LASSO \citep{grouplasso1}. The second scheme offers general regularization to prevent overfitting in the semiparametric dependence model by balancing between a linear Gaussian copula and nonlinear dependence.
    \item We create dedicated metrics for diagnosing conditional independencies in the resulting multivariate model.
    \item We offer methods to interpret the learned nonlinear conditional dependencies by converting them into varying local conditional pseudo-correlations.
\end{enumerate}

The remainder of this paper is structured as follows: In \autoref{theory:mctm}, we briefly review Multivariate Conditional Transformation Models and Normalizing Flows. \autoref{theory:extension_goals} introduces our model, the penalization scheme, and the conditional independence metric. In \autoref{sec:simstudy_vine_copulas}, we demonstrate the viability of our approach through numerical simulations with data generated from various vine copula structures. \autoref{sec:complex_distribution} applies our model to the MAGIC (Major Atmospheric Gamma-ray Imaging Cherenkov) telescope dataset. Lastly, in \autoref{sec:conclusion}, we summarize our work and draw conclusions. The appendix provides additional details: notation in \autoref{app:notation}, the mathematical proof for full conditional independence in vine copulas in \autoref{app:prof_theorem_rvine_ci}, computational specifics in \autoref{app:comp}, further simulation results in \autoref{app:vine_sim}, and application results in \autoref{app:magic}. 
We implemented the model in Python, and the corresponding code is publicly accessible at \href{https://github.com/MatthiasHerp/gtm}{Github}.

\section{Multivariate Conditional Transformation Models} \label{theory:mctm}

\subsection{Model Setup}

Multivariate Conditional Transformation Models (MCTMs), proposed by \cite{MCTM}, are a multivariate extension of Conditional Transformation Models \citep[CTMs,][]{CTM} that combine a Gaussian Copula with marginal CTMs. The model can be written as a transformation $\bfh(\bfy) = \bfz$ of the response $\bfy$ into a multivariate standard Gaussian space:
\begin{equation} \label{eq:mctm_model}
    \bfz = \bfh(\bfy) =
    \mathbf{\Lambda} \bfth(\bfy) = 
    \begin{bmatrix} 
    1 & 0 & \hdots & 0  \\
    \lambda_{2,1} & 1 & \hdots & 0 \\
    \vdots & \ddots & \ddots & \vdots \\
    \lambda_{J,1} & \cdots & \lambda_{J,J-1} & 1
    \end{bmatrix} \\
    \begin{bmatrix}
    \thh_1(y_{1}) \\
    \thh_2(y_{2}) \\
    \vdots \\
    \thh_J(y_{J})
    \end{bmatrix} \\
    \sim \calN(\mathbf{0},\mathbf{I}),
\end{equation}
where $\lambda_{r,c} \in \bbR$ for all $r=2, \dots, J$ and $c=1, \dots, J-1$. The matrix $\mathbf{\Lambda}$ provides a bijective mapping by virtue of its lower triangular structure with unit diagonal.
All marginal transformation functions $\thh_j: \bbR \rightarrow \bbR$, $j = 1, \dots, J$, are strictly monotonically increasing in the respective $y_j$, thereby forming bijective mappings. 
In the existing literature, $\thh_j(y_{j})$ is typically parameterised using $P_j$ basis function evaluations $a_{p}(y_j)$ and corresponding parameters $\bsupsilon_j$, such that we can write
\begin{equation}
   \thh_{j}(y_{j}) = \mathbf{a}(y_{j})^T \boldsymbol{\upsilon}_{j} = \sum_{p=1}^{P_{j}} a_{p}(y_{j}) \upsilon_{j,p}.
\end{equation}
For the bases $a_{p}$, \cite{CTM} and \cite{MCTM} have used Bernstein polynomials, while \cite{bayesian_CTM} employed B-Splines.
To ensure a strictly monotonically increasing function, the constraint $\upsilon_{j,1} < \upsilon_{j,2} < \hdots < \upsilon_{j,P_j}$ is sufficient for both basis choices, as demonstrated for Bernstein polynomials by \cite{MLT} and for B-Splines by \cite{bayesian_CTM}.
To enforce this restriction without resorting to constraint optimisation, as proposed by \cite{MLT}, , we adopt the approach of \cite{WoodConstraint}. We optimize an unrestricted parameter vector $\boldsymbol{\theta}_j = (\theta_{j,1}, \theta_{j,2}, ... \theta_{j,P_j})^\sfT$ which is subsequently transformed into the restricted parameter vector $\boldsymbol{\upsilon}_j = (\upsilon_{j,1}, \upsilon_{j,2}, ... \upsilon_{j,P_j})^\sfT$. The transformation is defined as $\upsilon_{j,p} = \theta_{j,1} + \sum_{\tilde{p}=2}^p \exp(\theta_{j, \tilde p})$, such that $\theta_{j,2}, \dots \theta_{j,P_j}$ can be understood as log increments in $\upsilon_{j,1}, \dots, \upsilon_{j, P_j}$.
For B-spline bases, \cite{bayesian_CTM} prove that the resulting restricted transformation function is monotonically increasing, as detailed in Theorem 2.1.

MCTMs are termed conditional because both the marginal transformations $\thh_j(y|\bfx)$ and the joint dependency defining matrix $\mathbf{\Lambda}(\bfx)$ can be conditioned on a set of covariates $\bfx$.
In this paper, we omit the conditioning on covariates for the GTM, leaving this aspect for future research.
Furthermore, we refer to the collection of marginal transformations in $\bfth(\bfy)$ as the \textit{transformation layer}. Additionally, we label $\mathbf{\Lambda}$ as the \textit{decorrelation layer}, as it removes inter-variable correlations within the transformed space $\bfth(\bfy)$, resulting in the standard Gaussian latent space.
\subsection{Interpretation as a Gaussian Graphical Model} \label{theory:mctm_ggm}
An intriguing characteristic of the MCTM is its relationship to a Gaussian Graphical Model (GGM) with arbitrary marginals.
The critical point is that the intermediate latent $\mathbf{\tilde{h}}(\bfy) = \bftz$ can be understood as following a multivariate Gaussian distribution with covariance matrix $\bsSigma = \bsLambda^{-1}\bsLambda^{-\sfT}$, given that $\tilde \bfh(\bfy) = \bsLambda^{-1} \bfz$ where $\bfz \sim \calN(\bfzero, \bfI)$.
Since $\tilde \bfh(\bfy)$ consists of independent element-wise transformations $\tilde h_j(y_j)$, it does not capture dependence among the elements of $\bfy$. 
Consequently, dependence is exclusively modeled by the linear mapping $\mathbf{\Lambda}$ in the transformed latent space $\bftz$, allowing us to consider the MCTM as applying a GGM to $\bftz$.
For GGMs, the precision matrix $\bfP =\mathbf{\Sigma}^{-1} = \bsLambda^{\sfT}\bsLambda$ is crucial as its elements encode the conditional correlations between all variable pairs, given all other variables. In particular, if the element $\bfP_{[r,c]} = \bfP_{[r,c]} = 0$, this implies conditional independence of $\tz_r$ and $\tz_c$ given all other elements in $\bftz$.
In turn, as the marginal transformations do not capture any interdependence, $\bfP_{[r,c]} = \bfP_{[r,c]} = 0$ also implies that $y_r$ and $y_c$ given all other elements in $\bfy$.
Based on the pairwise conditional independencies evident in $\bfP$, an undirected graphical model can be created, with edges representing the strength of the full conditional dependencies between nodes representing the variables. 
\subsection{Likelihood-based Inference} \label{theory:mctm_inference}
\cite{MCTM} follow \cite{MLT} in using maximum likelihood to estimate the model parameters. 
The maximum likelihood inference hinges on the tractable log-density of the data $\log f(\bfy)$. Due to the bijective setup of $\bsLambda \tilde \bfh(\bfy)$, $\log f(\bfy)$ can be obtained through an application of the change of variables theorem, and is given by
\begin{equation} \label{eq:transformation_likelihood}
\log f(\bfy)
=
\log \phi\bigl( \bsLambda \tilde{\mathbf{h}}(\bfy)\bigr) + \sum_{j=1}^{J} \log \left|\frac{\partial \tilde{h}_j(y_{j} | \bstheta_j)}{\partial y_{j}} \right|,
\end{equation}
where $\phi$ denotes the multivariate standard Gaussian density. 
The estimands in \eqref{eq:transformation_likelihood} are the non-fixed parameters in $\bsLambda$ and the parameter vectors $\bstheta_1, \dots, \bstheta_J$.
Due to its lower triangular structure and unit diagonal, the Jacobian determinant of $\bsLambda$ simplifies to one and can be omitted in \eqref{eq:transformation_likelihood}, leaving only the univariate derivatives of the transformation layer to consider. The log-likelihood for a sample of size $N$ arranged in a matrix as $\bfY = [\bfy_1, \dots, \bfy_N]$ is then expressed as $\ell(\bfY) = \sum_{i=1}^N \log f(\bfy_i)$. We obtain parameter estimates by maximizing $\ell(\bfY)$ with respect to the estimands. 
\subsection{Synthetic Sampling} \label{theory:mctm_synthetic_samples}
Generating synthetic samples from an MCTM is straightforward due to the invertibility of both the transformation layer and the decorrelation layer. To generate a new sample, draw $\bfz^* \sim \calN(\bfzero, \bfI)$ from a standard Gaussian distribution and apply the inverse mapping, yielding a new sample $\bfy^* = \tilde \bfh^{-1}(\bsLambda^{-1}\bfz^*)$ from the target distribution. The inverse $\tilde \bfh^{-1}$ can be obtained by numerically inverting the independent transformation functions constituting $\tilde \bfh$ as detailed in the appendix \autoref{alg:inverse_transformation}.
\section{Graphical Transformation Model} \label{theory:extension_goals}
As stated in \autoref{sec:introduction}, the principal objective of our work is to improve the MCTM's ability to represent complex dependencies among the dimensions of $\bfy$. Essentially, we aspire to surpass the limitations of the Gaussian Copula while maintaining a tractable likelihood and a graphical interpretation, emphasizing the encoding of conditional independencies in $f(Y)$.
Normalizing flows are central to our extension, so we will provide a brief introduction to them in \autoref{theory:normalizing_flow}. Subsequently, we will introduce our extension in \autoref{theory:gtm}.
As part of this development, we explore penalization schemes in \autoref{theory:penalties}, along with defining an exact metric for conditional independence in \autoref{theory:ci_metric}.
\subsection{Normalizing Flows} \label{theory:normalizing_flow}
Normalizing flows, as introduced by \cite{Nice_Dinh_2015} in the machine learning literature, are probabilistic models designed to learn intricate distributions in high-dimensional spaces. They achieve this by representing a complex target distribution as a sequence of simple bijective transformations to a standard Gaussian space.
From a statistical point of view, normalizing flows can be seen as transformation models. 
The primary distinction in the application of normalizing flows is their use of a series of transformations, $g_1, g_2, \ldots, g_L$, applied sequentially. The resulting forward pass of the joint transformation in a normalizing flow model is expressed as:
\begin{equation} \label{eq:normalizing_flow}
    \mathbf{z} = \bigl(g_L \circ g_{L-1} \circ \cdots \circ g_2 \circ g_1\bigr)(\bfy) \sim \mathcal{N}(\mathbf{0},\mathbf{I}).
\end{equation}
As in transformation models, each layer $g_l$ for $l =1, 2, \ldots, L$ must be a differentiable, bijective function to ensure a tractable likelihood.
Similar to transformation models, normalizing flows allow for recovering the original variable $\bfy$ by applying the composition of inverse transformations in reverse order to the standard Gaussian vector $\bfz$: $\bfy= (g_1^{-1} \circ \cdots g_L^{-1})(\bfz)$. Conversely, drawing a new $\bfz^*$ enables the generation of synthetic samples.
The key aspect of setting up a normalizing flows model lies in defining the mapping functions $g_1,g_2,...,g_L$. \cite{NF_Review_Kobyzev, DeepMind_NF_Review} offer comprehensive reviews of normalizing flows, detailing their applications and the various types of mapping functions, referred to as flows or layers.

In the terminology of normalizing flows, the decorrelation layer of the MCTM can be viewed as a type of \textit{coupling layer}.
Generally, coupling layers are defined as in Equation \eqref{eq:generalized_coupling_layer}, with the split input $\tz=(\tz_A,\tz_B)$ and output $z=(z_A,z_B)$, a coupling function $g$, and a conditioner $c$.
\begin{equation} \label{eq:generalized_coupling_layer} 
\begin{aligned}
z_A &= \tz_A \\
z_B &= g(\tz_B; c(\tz_A))
\end{aligned}
\end{equation}
The only requirement is that $g$ must be invertible given $c(\tz_A)$.
The decorrelation layer of the MCTM is a coupling layer with $g(\tz_B; c(\tz_A)) = \tz_B + c(\tz_A)$ and $c(\tz_A)= \lambda_{\tz_B, \tz_A} \cdot \tz_A$ in the bivariate case. 
For higher dimensions, the MCTM essentially acts as a coupling layer with the maximum number of splits along the data dimension, rather than a split in two, which is typical for normalizing flows.
This type of flow has also been introduced as Masked Autoregressive Flows \cite[MAF,][]{MaskedAutoregressiveFlow}.
\subsection{Graphical Transformation Model} \label{theory:gtm}

From normalizing flows, we adopt two concepts.
First, we redefine the $\lambda_{r,c}$ entries of $\mathbf{\Lambda}$ as functions rather than constants, making them dependent on the variable they multiply, resulting in $\lambda_{r,c}(\tz_c)$. 
This forms an additive coupling flow, as described in Equation \eqref{eq:generalized_coupling_layer} with the conditioner $c_{r,c}(\tz_c) = \lambda_{r,c}(\tz_c) \cdot \tz_c$. 
Since we remain within the class of coupling layers, the resulting decorrelation layer is a valid flow, thus a bijective function that meets the requirements for maximum likelihood estimation.
A convenient property, as noted by \cite{RNVP_Dinh_2017}, is that the conditioner $c_{r,c}(\tz_c)$ and hence $\lambda_{r,c}(\tz_c)$ need not be invertible nor is it necessary to compute their derivative in likelihood inference, due to the triangular matrix design of $\mathbf{\Lambda}$. 
Thus, we have full flexibility in choosing the functional form of the conditioner. 
We choose to use a spline as they are 
highly flexible, simple to evaluate and penalize to avoid overfitting.
In particular, we use a B-Spline. 
The second concept borrowed from normalizing flows is the idea of creating a sequence of alternating layers to enhance model flexibility. 
An alternating pattern can be created by applying every second decorrelation layer to a flipped input.
This is fundamental in normalizing flows, ensuring that variable dependencies can be in any direction \citep{RNVP_Dinh_2017}.
To flip the inputs we use the exchange matrix $\mathbf{F}$ with ones on its anti-diagonal. It is 
defined as $\mathbf{F}_{r,c} = 1$ if $r + c = J + 1$ and $\mathbf{F}_{r,c} = 0$ otherwise.

By combining these two concepts, we attain Equation \eqref{eq:decorr_layer_nf} and Equation \eqref{eq:decorr_layers_nf_fct_def}, which define the sequence of functions applied in Equation \eqref{eq:normalizing_flow}. 
The former defines a single coupling decorrelation layer $\mathbf{\Lambda}_l$.
The latter utilizes $\mathbf{\Lambda}_l$ and $\mathbf{F}$ to create the complete dependency structure of the model as a sequence of these layers, assuming $L$ is an even number, thus including flipping in the last layer.
\begin{equation}
\label{eq:decorr_layer_nf}
    \mathbf{\Lambda}_l(\bftz_{l-1}) = \begin{bmatrix} 
  1 & 0 & \hdots & 0   \\
  \lambda_{2,1,l}(\bftz_{1,l-1}) & 1 & \hdots & 0 \\
  \vdots & \ddots & \ddots & \vdots \\
  \lambda_{J,1,l}(\bftz_{1,l-1}) & \lambda_{J,2,l}(\bftz_{2,l-1}) &  \hdots & 1 \\
  \end{bmatrix}
\end{equation}
\begin{equation} 
\label{eq:decorr_layers_nf_fct_def}
    \mathbf{h}_l(\bftz_{l-1}) =
    \begin{cases}
    \mathbf{F} \mathbf{\Lambda}_l( \mathbf{F} \bftz_{l-1})  \mathbf{F} \bftz_{l-1} & 
    \quad l \mod 2 = 0 \\
    \mathbf{\Lambda}_l(\bftz_{l-1}) \bftz_{l-1} & 
    \quad \textrm{otherwise}
    \end{cases}
\end{equation}
Each nonlinear decorrelation layer's inversion, akin to any coupling layer, can be calculated iteratively by solving the equations from top to bottom, as outlined in \autoref{alg:inverse_decorrelation}. 
According to \cite{RNVP_Dinh_2017}, a minimum of three layers is recommended, as this number allows each variable to affect all others, conditioned on all others, effectively mitigating the influence of variable ordering. In practice, the number of layers is linked to training efficiency, with three being the lower limit for a hyperparameter that can be increased for the model to better approximate complex distributions.

A crucial factor in choosing this type of flow is its preservation of the MCTM decorrelation layer's structure. 
Despite the matrix entries varying, each layer consists of a linear transformation $\mathbf{\Lambda}(\bftz)$ given the marginally transformed data $\mathbf{\bfy}$.
This becomes apparent by Equation \eqref{eq:lambda_matrix_normalizing_flow}, which defines the $\mathbf{\Lambda}$ matrix of the Equation \eqref{eq:decorr_layers_nf_fct_def}:
\begin{equation} 
\label{eq:lambda_matrix_normalizing_flow}
    \mathbf{\Lambda}(\bftz) =
    \prod_{l=L}^1
    \begin{cases}
    \mathbf{F} \mathbf{\Lambda}_l( \mathbf{F} \bftz_{l-1}) \mathbf{F} & 
    \quad l \mod 2 = 0 \\
    \mathbf{\Lambda}_l(\bftz_{l-1}) & 
    \quad \textrm{otherwise}
    \end{cases}
\end{equation}
with $\prod_{l = L}^{1} \mathbf{\Lambda}_l := \mathbf{\Lambda}_L \mathbf{\Lambda}_{L-1} \cdots \mathbf{\Lambda}_1$.
For example, when $L=3$, the matrix is given by $\mathbf{\Lambda}(\mathbf{\bftz}) =
\mathbf{\Lambda}_3(\bftz_{2})
\mathbf{F} \mathbf{\Lambda}_2( \mathbf{F} \bftz_{1}) \mathbf{F}
\mathbf{\Lambda}_1(\bftz)$.
Although the intermediate latent spaces $\bftz_1, \bftz_2$ serve as inputs for functions, by iteratively substituting definitions of prior flow layers, all functions can be expressed as a nested function of the transformed latent space $\bftz$ to derive:
$$\mathbf{\Lambda}(\bftz) =
\mathbf{\Lambda}_3(
\mathbf{F} \mathbf{\Lambda}_2( \mathbf{F} \mathbf{\Lambda}_1(\bftz) \bftz) \mathbf{F} \mathbf{\Lambda}_1(\mathbf{\bftz}) \bftz )
\mathbf{F} \mathbf{\Lambda}_2( \mathbf{F} \mathbf{\Lambda}_1(\bftz) \bftz) \mathbf{F}
\mathbf{\Lambda}_1(\bftz)
$$ in the $L=3$ case.
Thus, given $\bftz$, all triangular matrices of the coupling layers, and hence the joint conditional linear transformation $\mathbf{\Lambda}(\bftz)$, can be computed.

In turn we can then compute $\bfP(\bftz) = \bsLambda(\bftz)^{\sfT}\bsLambda(\bftz)$.
Since $\bfP(\bftz)$ depends on $\bftz$, the GTM no longer models a Gaussian Copula.
To differentiate $\bfP(\bftz)$, from the Gaussian MCTM, we call it the local pseudo-precision matrix.
local because it reflects the dependence structure at a specific point $\bftz$, and pseudo because it is not a precision matrix in the Gaussian sense.
Its off-diagonal $p_{r,c,n}$ do not just vary across pairs $r,c$ but also across observations $n$ and potentially depend on all dimensions of $\bftz$ and therefore $\bfy_n$. 
In other words, for two different observations $\mathbf{y}_1$ and $\mathbf{y}_2$, the local pseudo precision matrix entries are not equal, i.e., $p_{r,c,1} \neq p_{r,c,2}$, which implies that correlations can also differ, $\rho_{r,c,1} \neq \rho_{r,c,2}$. Hence, we refer to $\rho_{r,c,n}$ as local conditional pseudo-correlations.
The $\rho_{r,c,n}$ offer great interpretational advantages, as they can assist in understanding nonlinear and even non monotonic conditional dependencies, as demonstrated in \autoref{fig:sign_conditional_corr_h} of the application in \autoref{sec:complex_distribution}
Furthermore, both the $\rho_{r,c,n}$ and $p_{r,c,n}$ are also closely linked to the conditional independencies, as we discuss in both in the application and the simulation study in \autoref{sec:simstudy_vine_copulas}.
Therefore, we can leverage the $p_{r,c,n}$ to construct an approximate conditional independence penalty, which we will elaborate on in \autoref{theory:penalties}.

\subsection{Penalization Scheme} \label{theory:penalties}
The modelling choices are essential for the interpretation of our two penalization schemes. 
First, as is common practice in the statistical literature for models employing splines, we apply ridge penalties to the derivatives of every spline in each decorrelation layer, namely $\lambda_{r,c,l}$ with there respective parameter vector $\mathbf{\theta}_{l,r,c}$ (\cite{pspline}). 
We penalize both the first and second derivatives.
The penalty term, that is to be added to the log-likelihood 
and includes hyperparameters $\tau_{1}$ and $\tau_{2}$, is given by:
\begin{equation}
    \text{\footnotesize Spline-Penalty} = 
     \tau_{1} \mathbf{1}^T \left( \mathbf{D}_1 \mathbf{\theta} \right)^2  + 
     \tau_{2}  \mathbf{1}^T \left( \mathbf{D}_2 \mathbf{\theta} \right)^2 
\end{equation}
where $\mathbf{D}_1$ and $\mathbf{D}_2$ are first and second order differencing matrices. 
For clarity and brevity, , we suppress indices that signify summing the penalty across all splines $\lambda_{r,c,l}$ in each layer $\mathbf{\Lambda}_l$.
If $\tau_{1} \rightarrow \inf$, any differences in $\mathbf{\theta}_{l,r,c}$ are heavily penalized, causing each $\lambda_{r,c,l}$ to reduce to a constant.
Consequently, every $\mathbf{\Lambda}_l$ becomes a linear transformation, and thus the product of all layers $\mathbf{\Lambda}$ results in an overall linear transformation.
Therefore, $\tau_{1}$ regulates how much the normalizing flow is regularized towards the baseline linear MCTM.
In this manner $\tau_{1}$ mediates between extreme nonlinearity and the Gaussian copula.
The roles of $\tau_{2}$ is simply to smooth the splines.
Regarding the splines in the transformation layer, when employing a sparse basis, an additional penalty is unnecessary. Furthermore, a first derivative penalty is unwarranted since the splines are required to be monotonically increasing. For complex marginal data, such as the MAGIC data analyzed in \autoref{sec:complex_distribution}, a large basis may be considered in conjunction with a ridge penalty on the second derivatives for smoothing. In these models, we introduce the penalty hyperparameter $\tau_{4}$ and incorporate the spline penalty into the penalized likelihood accordingly.
Building on penalization towards a Gaussian copula, we apply a second penalization scheme derived from Gaussian Graphical Models (GGM).
Following \cite{Graph_Lasso} and \cite{graph_lasso_Tibshirani}, we employ a LASSO \citep{lasso} penalization to the off-diagonal entries of the precision matrix.
In this context, we apply the penalty to the local pseudo-precision matrix $\bfP(\bftz)$ to approximately encourage a sparse conditional independence structure.
If $\tau_{1}$ does not constrain the model to a Gaussian Copula, $\bfP(\bftz)$ of the GTM is dependent on $\bfy_n$. 
Therefore, a simple LASSO penalization does not necessarily ensure sparseness in terms of conditional independence relationships across all observations $n$ within a variable pair, thereby concentrating zeros in specific $p_{r,c,-}$ entries.

Instead we utilize a group LASSO (\cite{grouplasso1, grouplasso2, grouplasso3, grouplasso5}) penalty, which is equivalent to applying the second norm across observations $n$ for each pair $r,c$:
\begin{equation} 
\label{eq:penalty_precision_matrix}
    \text{\footnotesize LASSO-Penalty} =  \tau_3 \sum_{r \neq c, r > c} \left(\sum_{n=1}^N p_{r,c,n}^2\right)^{0.5}
\end{equation}
Implementing this form of penalization encourages the emergence of a sparse GTM in the sense of concentrating zeros in certain $p_{r,c,-}$ entries.
However, this penalization scheme only leads to conditional independence if the nonlinearity penalty $\tau_1$ confines the model to a Gaussian Copula. 
For the nonlinear scenario, the penalty serves merely as an approximation of a conditional independence penalty. 
A detailed explanation is provided in \autoref{appendix:challenge_independence_identification}.

As an alternative to the standard LASSO penalty, we implement an adaptive LASSO penalty following \cite{adaptive_lasso}. This involves first training the GTM without any LASSO penalty to establish weights $w_{r,c} = \frac{1}{N}\sum_{n=1}^N |p_{r,c,n}|$ based on $\bfP(\bftz)$. 
In the subsequent step, we retrain the model with the adaptive LASSO penalty:
\begin{equation} 
\label{eq:adaptive_penalty_precision_matrix}
    \text{\footnotesize Adaptive-LASSO-Penalty} =  \tau_3 \sum_{r \neq c, r > c} \frac{1}{w_{r,c}} \left(\sum_{n=1}^N p_{r,c,n}^2\right)^{0.5}
\end{equation}
The adaptive LASSO's rationale is that by weighting the penalty, smaller pairwise average local pseudo-precision matrix entries are penalized more heavily than larger ones. 
This approach may enhance the focus of the penalty on identifying conditional independencies while reducing bias with respect to conditionally dependent pairs.

\subsection{Conditional Independence Metric} \label{theory:ci_metric}
\newcommand{\fuv}{f(\bfy_{u, v} \mid \bfy_{- u, v})}
\newcommand{\fuvind}{f_{u \perp v}(\bfy_{u, v} \mid \bfy_{- u, v})}
\newcommand{\fuvs}{f(\bfy_{u, v}^s \mid \bfy_{- u, v}^s)}
\newcommand{\fuvinds}{f_{u \perp v}(\bfy_{u, v}^s \mid \bfy_{- u, v}^s)}
As the conditional independence in the GTM cannot be simply defined by zero entries in $\bfP(\bftz)$, we propose to evaluate it based on a likelihood ratio.
This method provides precise metrics for assessing conditional independence in the distribution defined by the GTM.
The rationale begins with the definition of conditional independence: two random variables $y_1$ and $y_3$ are considered conditionally independent given a third variable $y_2$—denoted as $y_1 \perp y_3 \mid y_2$—if and only if their joint conditional density $f(y_1,y_3|y_2)$ can be factored into the product of marginal conditional densities, $f(y_1, y_3|y_2) = f(y_1|y_2)f(y_3|y_2)$. Thus, we can measure the closeness of a distribution to conditional independence by comparing $f(y_1, y_3 \mid y_2)$ to $f_{\perp}(y_1, y_3 \mid y_2) = f(y_1|y_2) f(y_3|y_2)$. The closer the ratio of the two densities is to one, or equivalently, the nearer their log differences are to zero, the closer the distribution defined by $f(\bfy)$ is to conditional independence $y_1 \perp y_3 \mid y_2$.

We apply this principle to our $J$-dimensional GTM and propose calculating two common likelihood ratio-based measures: the Kullback-Leibler Divergence ($\kld$) and the normalized Integrated Absolute Error ($\iae$). 
In information theory, the $\kld$ is understood as quantifying the information loss when using the GTM with the independence assumption instead of the learned GTM. 
The $\kld$ does not have an upper limit, so although it is effective for ranking pairs by proximity to conditional independence, it is challenging to interpret in absolute terms—determining which values between $f$ and $f_{\perp}$ indicate evidence of conditional independence versus those indicating conditional dependence.
As a complementary measure, we also compute the Integrated Absolute Error ($\iae$), generally defined for two densities as $
\iae(f_1,f_2) = \frac{1}{2} \int | f_1(y) - f_2(y) | \mathrm d y
$. 
The normalized $\iae$ falls within the range $[0,1]$ after normalization by 2, which can be interpreted as the percentage of non-overlapping probability mass between the two densities.
Thus, it represents the percentage error in probability mass when approximating the learned distribution under the assumption of conditional independence. 
Applied researchers can determine a suitable threshold below which pairs can be considered conditionally independent, based on the context of their application.

To elaborate on approximating the $\kld$ and the $\iae$ via sampling, we first define the model density under the full conditional independence assumption between two elements of $\bfy$ indexed by $u$ and $v$, $u \neq v$, given $\bfy_{- u, v}$, the subset of all elements of $\bfy$ indexed by  $\{1, 2, ..., J\} \setminus \{u, v\}$ as $\fuvind = 
f(y_u \mid \mathbf{\bfy}_{- u, v}) 
\times
f(y_v \mid \mathbf{\bfy}_{- u, v})$.
We compare this distribution to the joint density of $u$ and $v$ given $\bfy_{- u, v}$ directly implied by the model, $\fuv$.
We first define the $\kld$ and $\iae$ metrics for a fixed conditioning set, termed as the \textit{local} versions of these metrics:
\begin{align*}
\kld^{\text{local}}_{f, f_\perp | \bfy_{- u, v}} &= 
\int 
\fuv \times
\log \left[ 
\frac{\fuv}{\fuvind}
\right] \,
d\bfy_{u,v} \\
\iae_{u,v | - u, v}^{\text{local}} &= 
\frac{1}{2}
\int 
|\fuv - \fuvind| \,
d\bfy_{u,v}
\end{align*}
The $\kld_{u,v | - u, v}^{\text{local}}$ and $\iae_{u,v | - u, v}^{\text{local}}$only measure the conditional independence given a fixed conditioning set $\bfy_{- u, v}$.
To generalize them over all possible conditioning set values, we take the expectation across these sets:
\begin{align*}
\kld_{u,v | - u, v} &= 
\int 
f(\bfy_{- u, v})
\left(
\int 
\fuv \times
\log \left[ 
\frac{\fuv}{\fuvind}
\right] \,
d\bfy_{u,v}\right) \,
d\bfy_{- u, v}\\
\iae_{u,v | - u, v} &=
\frac{1}{2}
\int 
f(\bfy_{- u, v})
\left(
\int 
|\fuv - \fuvind| \,
d\bfy_{u,v} \right) \,
d\bfy_{- u, v}
\end{align*}
To approximate these integrals, we generate $S$ random draws $(\bfy^1, \bfy^2, \ldots, \bfy^S)$ from our learned distribution $f(\bfy)$, where each $\bfy^s \in \bbR^J$ represents one $J$-dimensional realization of the estimated distribution of the multivariate random variable $Y$ for $s=1, \ldots, S$. We then evaluate $\fuvs$ and $\fuvinds$, averaging the metrics across samples:
\begin{align*}
\kld_{u,v | - u, v}& = 
\frac{1}{S}
\sum_{s=1}^S
\frac{
f(\bfy_{- u, v}^s)
\fuvs
}{f(\bfy^s)}
\times
\log \left[ 
\frac{\fuvs}{\fuvinds}
\right] \\
&= 
\frac{1}{S}
\sum_{s=1}^S
\log \left[ 
\frac{\fuvs}{\fuvinds}
\right] \\
\iae_{u,v | - u, v} &=
\frac{1}{2 S}
\sum_{s=1}^S
\frac{
f(\bfy_{- u, v}^s)
}{f(\bfy^s)}
\times
|\fuvs - \fuvinds| 
\end{align*}
Here, the $\kld_{u,v | - u, v}$ simplifies to a log-likelihood ratio across samples while the $\iae_{u,v | - u, v}$ requires weighting by $f(\bfy_{- u, v}^s)/f(\bfy^s)$ to account for the conditioning set and sampling probabilities. 
The $\iae_{u,v | - u, v}$ provides dual interpretations: it represents the expected percentage error in probability mass across all conditioning sets, when assuming conditional independence for a given pair, or it can be seen as the integrated absolute error between the full distribution with and without the conditional independence assumption for the pair: $\iae_{u,v | - u, v} = \iae(f(\bfy), f_{u \independence v | - u,v}(\bfy))$. 
A particularly appealing feature of this approach is the ability to compute approximations for both metrics from the same set of random draws from the fitted model, thereby saving computational resources.

In Algorithm \ref{alg:llr_metrics}, we provide additional details on the computations. All calculations are performed using the estimated joint probability density $\hat f(\bfy)$, substituting in the maximum likelihood estimates for all model parameters; however, to simplify notation, we continue to use $f(\bfy)$. Since we cannot analytically compute the one- and two-dimensional integrals necessary to obtain the conditional densities, we approximate them using Gauss-Legendre Quadrature, denoted by the function $\text{GLQ}$. 
Additionally, it should be noted that the metrics can be computed in either the observed space of $Y$ or the latent space of $\tilde{Z}_0$. This is feasible because the marginal transformations $\tilde Z_0 = \bfh(Y)$ do not capture dependencies, implying that conditional independence among elements of $\tilde Z_0$—and consequently in $f(\tilde Z)$—is equivalent to conditional independence among elements of $Y$ in $f(Y)$. 
Calculating the metrics in $\tilde{Z}_0$ can offer the benefit of requiring fewer knots for numerical integration, leading to greater precision and numerical stability, especially in complex marginal models.
\begin{algorithm}
\caption{Approximate Likelihood Ratio Based Metrics for Full Conditional Independence}
\label{alg:llr_metrics}
\textbf{Input:} Samples $\begin{bmatrix} \bfy^1 & \bfy^2 & \dots & \bfy^S \end{bmatrix}^T$ from the trained GTM with the probability density $f(\bfy)$.\\
\textbf{Output:} $\iae_{u,v}$ $\kld_{u,v}$ for all pairs $(u,v)$.
\begin{algorithmic}
\FORALL {Samples $\bfy^s$ for $s \in [1,2,...,S]$}
    \STATE - Compute the likelihood  $f(\bfy^s)$ 
    \FORALL {pairs $(u,v) \in \{1, 2, \dots, J\} \times \{1, 2, \dots, J\}, u \neq v$}
        \STATE - Compute the marginal likelihood of the conditioning set for each sample $s =1,...,S$:
        \[
        f(\bfy_{- u, v}^s) \stackrel{\text{GLQ}}{\approx}
        \iint f(\bfy^s) \, \mathrm dy_u \, \mathrm dy_v
        \]
        \STATE - Compute the joint conditional density implied by the fitted model for each sample:
        \[
        f(\bfy_{\{u, v\}}^s \mid \bfy_{- u, v}^s) = 
        \frac{f(\bfy^s)}{f(\bfy_{- u, v}^s)}.
        \]
        \STATE - Integrate out variable $u$ to compute the marginal conditional density of $v$ without conditioning on $u$:
        \[
        f(\bfy_{-u}^s) \stackrel{\text{GLQ}}{\approx} \int f(\bfy^s) \, dy_u.  \quad \quad \quad
        f(y_{v}^s \mid \bfy_{- u, v}^s) =  \frac{f(\bfy_{-u}^s)}{f(\bfy_{- u, v}^s)}
        \]
        \STATE - Integrate out variable $v$ to compute the marginal conditional density of $u$ without conditioning on $v$:
        \[
        f(\bfy_{-v}^s) \stackrel{\text{GLQ}}{\approx} \int f(\bfy^s) \, dy_v. \quad \quad \quad
        f(y_{u}^s \mid \bfy_{- u, v}^s) =  \frac{f(\bfy_{-v}^s)}{f(\bfy_{- u, v}^s)}
        \]
        \STATE - Compute the conditional dependence structure under the independence assumption:
        \[
        \fuvinds = 
        f(y_{u}^s \mid \bfy_{- u, v}^s) 
        \times
        f(y_{v}^s \mid \bfy_{- u, v}^s)
        \]
        \ENDFOR
    \ENDFOR
\STATE - Compute Likelihood Ratio Metrics:
\[
\kld_{u,v | - u, v} \approx 
\frac{1}{S}
\sum_{s=1}^S
\log \left[ 
\frac{\fuvs}{\fuvinds}
\right]
\]
\[
\iae_{u,v | - u, v} \approx 
\frac{1}{2 S}
\sum_{s=1}^S
\frac{
f(\bfy_{- u, v}^s)
}{f(\bfy^s)}
\times
|\fuvs - \fuvinds| 
\]
\end{algorithmic}
\end{algorithm}

\section{Simulation Studies} \label{sec:simstudy_vine_copulas}
\subsection{Data Generation and Model} 

To test the GTM, we create various data-generation scenarios using vine copulas, aiming to evaluate the GTM's effectiveness in learning the underlying data-generating distribution and identifying conditional independencies.
Vine copulas are a powerful tool for data generation, as they enable the modeling of complex nonlinear dependence structures hierarchically using pairwise dependencies. Due to this property, vine copula structures are often referred to as pair copula constructions (PCC).
Moreover, the R software implementation \cite{r_vinecopula_package} facilitates efficient data generation, likelihood computation, and even the sampling of random pair copula constructions.
However, a limitation of vine copulas in our application is that full conditional independencies are not immediately apparent from the PCC, since it is not defined in terms of full conditionals. Nevertheless, as stated formally in \autoref{theo:rvine_ci}, full conditional independencies can be identified in a vine copula, with proof is provided in \autoref{app:prof_theorem_rvine_ci} along with further details in \autoref{app:vine_copula}.
\begin{theorem} \label{theo:rvine_ci}
Given an arbitrary R-vine structure of $J$ dimensions, for a random variable $\mathbf{Y}$, with trees $[T_1,T_2,...T_j,T_{j+1},...,T_{J-1}]$, that conforms to the simplifying assumption.
If all pair-copulas in trees $[T_j,T_{j+1},...,T_{J-1}]$ are independence copulas,
then each pair of variables $Y_u,Y_v$ with $u,v \in [1,..,J]$ that has its pair copula within the trees $[T_j,T_{j+1},...,T_{J-1}]$ is fully conditionally independent $Y_u \independence Y_v | \mathbf{Y}_{-u,v}$.
\end{theorem}
For our simulation we focus on $D=10$ dimensional data and exclusively sample dependence copulas in trees $T_1, T_2, T_3$, ensuring that all pair copulas in $T_4,\ldots,T_9$ are set to the independence copula. 
As a result, $21$ out of $45$ pairs are fully conditionally independent.
For the vine structure, we use both the extreme cases of D-Vines and C-Vines as well as randomly sampled R-Vine structures.
As pair copulas we randomly choose from the Independence, Gaussian, T, Frank, Joe, Gumbel, and Clayton copulas with all possible rotations. 
We create three scenario groups for the R-Vine: a \textit{baseline}, a \textit{weak dependence}, and a \textit{positive dependence}.
For baseline scenarios, Kendall's $\tau$ is sampled from a uniform distribution in the interval $[0.3, 0.7]$, with the sign of dependence (positive or negative) assigned with equal probability; weak dependence scenarios use the interval $[0.1, 0.4]$; positive dependence scenarios employ positive dependence rotations or restrict Kendall's $\tau$ sampling to positive values, according to the copula.

We fit the data using a simple three-layer GTM with $40$ equidistant knots spanning $[-15, 15]$ in each spline in the decorrelation layers $\bsLambda_l$ and $15$ knots on the same grid in the marginal splines in $\tilde \bsh(\bfy)$. Unlike the real-world data applications in \autoref{sec:complex_distribution}, hyperparameter optimization regarding model architecture features was not conducted in the simulation study.

\subsection{Fitting the True Distribution}

To assess how well the model learns the true distribution, we compute the $\kld$ to the true density and compare it to alternative models. 
Specifically, we estimate a multivariate Gaussian density $\hat{f}_G$, the true vine copula density $\hat{f}_{VC}$ with the known pair copula construction (PCC), and the GTM $\hat{f}_{GTM}$. 
For the Gaussian density estimation, we employ the GGM Python package \cite{pypackage_ggm}. 
Using the known PCC, $\hat{f}_{VC}$ only estimates the parameters of the pair copulas.
To evaluate the efficiency of our approximate conditional independence penalty, we provide results for three different GTM conditional independence penalty schemes: without any LASSO penalty, with a LASSO penalty, with an Adaptive LASSO penalty.
For each model, we compute the $\kld$ to the true density of the data-generating distribution $f_{VC}$ to obtain: $\kld_G = \kld(f_{VC} \| \hat{f}_{G})$, $\kld_{VC} = \kld(f_{VC} \| \hat{f}_{VC})$ and $\kld_{GTM} = \kld(f_{VC} \| \hat{f}_{GTM})$.
We then compute how well the GTM approximates the true model on the scale from the Gaussian to the estimated true model in terms of $\kld$: 
\begin{equation}
    \text{rKLD} = \frac{ \kld_{GTM} - \kld_{VC}}{\kld_{G} - \kld_{VC}}
\end{equation}
A value of $\text{\footnotesize rKLD}=0$ indicates that the GTM performs just as good as an estimation of the true model with the oracle-knowledge of the true pair copula construction, while a value of $\text{\footnotesize rKLD}=1$ means that the GTM provides no benefit over a simple multivariate Gaussian model. Values $\text{\footnotesize rKLD} < 1$ indicate a benefit of the GTM compared to the Gaussian model, while values larger than $1$ mean that the GTM performs worse than the Gaussian model.
We approximate each $\kld$ using the test set.

\autoref{fig:xvine_kld} depicts $\text{\footnotesize rKLD}$ results across the different vine scenarios.
The results demonstrate that the GTM without penalty can better approximate the true underlying distribution than a Gaussian model down to training sample sizes of $250$ observations, but tends to perform worse than the Gaussian baseline as the sample size decreases to $125$ observations. In the R-vine-weak scenario, the GTM clearly outperforms the Gaussian with $1000$ observations, draws even with $500$, and loses to the Gaussian with $250$ and smaller numbers of observations -- this structure is especially challenging as weak dependencies in pair copulas are well approximated by a Gaussian.
The GTM with a simple LASSO penalty performs largely similar to the unpenalized GTM, indicating no clear advantage from the LASSO penalization. 
However, adding an adaptive LASSO penalty noticeably enhances model performance at 125 training samples, allowing the GTM to match or slightly outperform the Gaussian approximation for the D-vine, R-vine, and C-vine scenarios.
For the D-vine, C-vine, R-vine-weak, and R-vine-positive scenarios, we even find that the adaptive LASSO improves upon the unpenalized model for sample sizes of $250$, $500$ and even $1000$. The improvement is especially pronounced in the R-vine-weak scenario, that appeared to be especially challenging for the GTM without adaptive LASSO.
Overall, we conclude that the adaptive LASSO penalty enhances the GTM's ability to approximate the true underlying distribution and, when it fails to provide improvement, it does no harm.
We attribute the improved performance to a reduction in overfitting.
This is particularly true for more challenging scenarios where dependencies are difficult to detect, such as smaller sample sizes that are prone to overfitting. 
Additionally, the adaptive LASSO is beneficial in situations where the Gaussian assumption performs well, such as when effects are weak.
\begin{figure}[ht]
    \centering
    \begin{minipage}{0.32\textwidth}
        \centering
        \includegraphics[width=\textwidth]{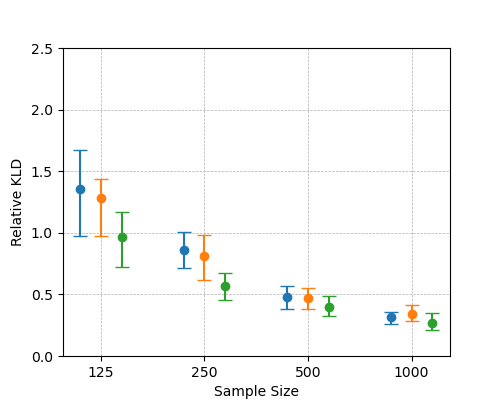}
        \subcaption{D-vine}
    \end{minipage}
    \begin{minipage}{0.32\textwidth}
        \centering
        \includegraphics[width=\textwidth]{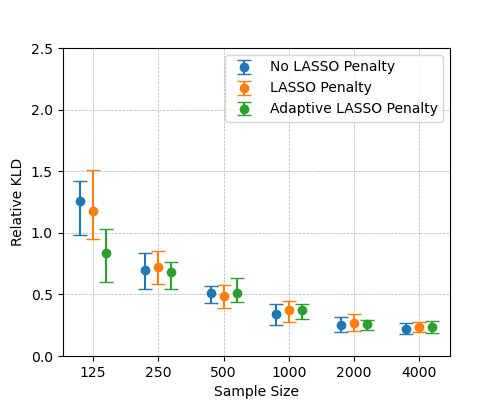}
        \subcaption{R-vine}
    \end{minipage}
    \begin{minipage}{0.32\textwidth}
        \centering
        \includegraphics[width=\textwidth]{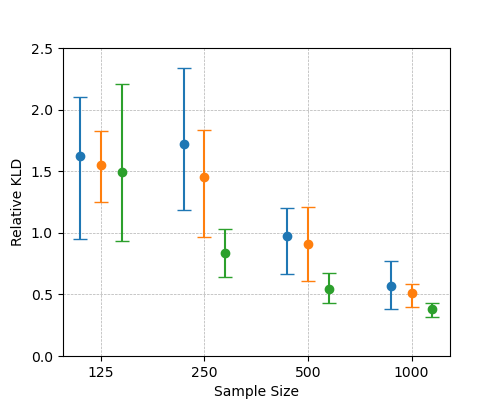}
        \subcaption{R-vine-Weak}
    \end{minipage}
    \begin{minipage}{0.32\textwidth}
        \centering
        \includegraphics[width=\textwidth]{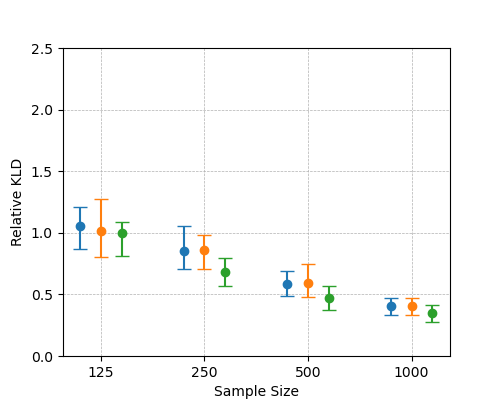}
        \subcaption{R-vine-Positive}
    \end{minipage}
    \begin{minipage}{0.32\textwidth}
        \centering
        \includegraphics[width=\textwidth]{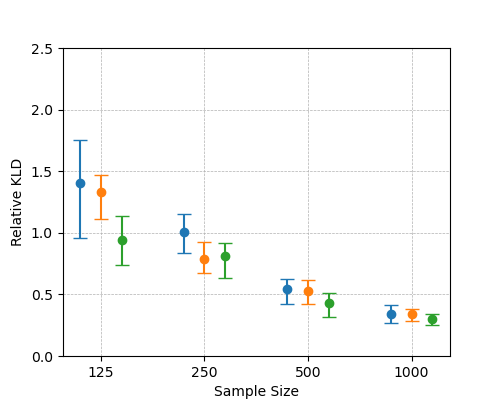}
        \subcaption{C-vine}
    \end{minipage}
    \caption{Across different vine scenarios, the figure depicts the relative $\kld$ of the GTM on the scale from the true model to a Gaussian approximation. A value of $0$ would mean that the GTM is as good as the true model with estimated parameters and a value of $1$ or larger means that the GTM is not better than a Gaussian approximation. The GTM performance for different training sample sizes is presented for the model without a LASSO penalty in blue, with a LASSO penalty in orange, and with an adaptive LASSO penalty in green. The $\kld$ is approximated via a $40.000$ observations large test sample. The dots represent the means and the whiskers the $20\%$ and $80\%$ quantiles across $30$ replications.}
    \label{fig:xvine_kld}
\end{figure}

\subsection{Identifying Conditional Independencies}

To evaluate the GTM's ability to identify conditional independencies, we compare it to a baseline GGM, by means of the Area under the Curve ($\auc$).
The $\auc \in [0,1]$ is a scalar measure of classification performance, representing the probability that a randomly chosen conditional dependence is ranked higher than a randomly chosen conditional independence. 
It is thus ideal to compare the validity of the ranking of conditional dependencies.
To rank conditional dependencies for the GGM we use the estimated conditional correlations.
As discussed in \autoref{theory:ci_metric}, we use the likelihood ratio based metrics for the GTM. 
Due to its better interpretability, we focus the report on the $\iae$ and depict its $\auc$ scores across scenarios in \autoref{fig:xvine_auc_iae}.
To compute the Gauss-Legendre-Quadratures required in \autoref{alg:llr_metrics} we use $20$ quadrature points and bound the quadrature at the borders of the transformation layer span.

We observe that the GTM generally outperforms the GGM, particularly for larger sample sizes, scenarios with stronger dependencies reflected by larger Kendall's $\tau$, and especially for data generated from C-vines.
The greater nonlinearity of dependencies, due to larger Kendall's $\tau$ leading to pronounced effects like increased asymmetric tail dependencies, together with the challenging nature of dependencies in C-vines, increase the GTM's advantage over the GGM. Remarkably, this trend is consistent across all three GTM variants, even the unpenalized one. Hence, these improvements can largely be attributed to the GTM’s enhanced capability to learn the underlying distribution more accurately.

When comparing the different GTM variants, we find no notable improvement by adding a LASSO penalty.
However, the adaptive LASSO penalty does offer enhancements, especially for small sample sizes of $125$ and $250$, with D-vines showing particularly consistent benefits.
The R-vine-weak scenario with $125$ training samples—characterized by near-linear dependencies due to smaller Kendall’s $\tau$ in the pair copulas—is the sole scenario where the GTM with the adaptive LASSO does not surpass the GGM in terms of identifying independence via the $\iae$.

In the appendix, we present the $\auc$ of the GTM for different metrics: for the $\kld$ in \autoref{app_fig:xvine_auc_kld}, for the mean absolute local pseudo-precision matrix entry $p_{r,c,-}$ in \autoref{app_fig:xvine_auc_pmatrix} as well as for the mean absolute local pseudo-conditional correlation $\rho_{r,c,-}$ in \autoref{app_fig:xvine_auc_condcorr}.
Across all evaluated metrics, we consistently find that the GTM outperforms the GGM in identifying conditional independencies, particularly for larger sample sizes, strong dependencies, and complex C-vine structures.
We also observe consistently that an adaptive LASSO penalty enhances performance for small sample sizes, whereas a simple LASSO penalty does not.
Lastly, we find that $p_{r,c,-}$ is competitive, and in some instances, even slightly better in terms of $\auc$ than the likelihood ratio metrics.
This is promising, as it indicates that our approximate penalty scheme effectively targets conditional independencies. 
\begin{figure}[ht]
    \centering
    \begin{minipage}{0.32\textwidth}
        \centering
        \includegraphics[width=\textwidth]{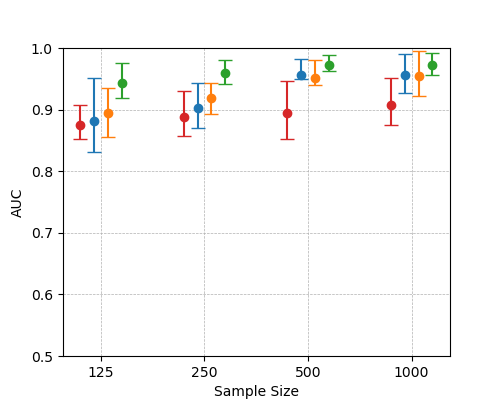}
        \subcaption{D-vine}
    \end{minipage}
    \begin{minipage}{0.32\textwidth}
        \centering
        \includegraphics[width=\textwidth]{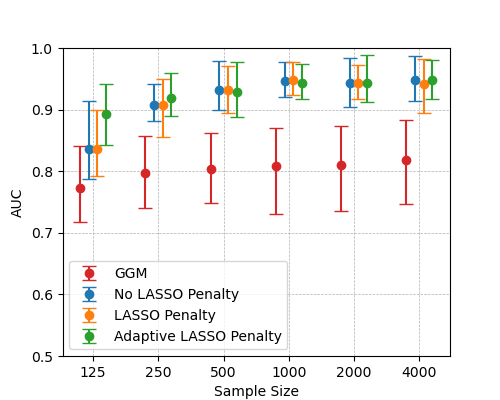}
        \subcaption{R-vine}
    \end{minipage}
    \begin{minipage}{0.32\textwidth}
        \centering
        \includegraphics[width=\textwidth]{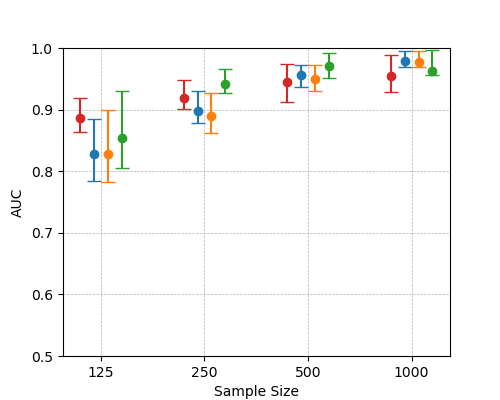}
        \subcaption{R-vine-Weak}
    \end{minipage}
    \begin{minipage}{0.32\textwidth}
        \centering
        \includegraphics[width=\textwidth]{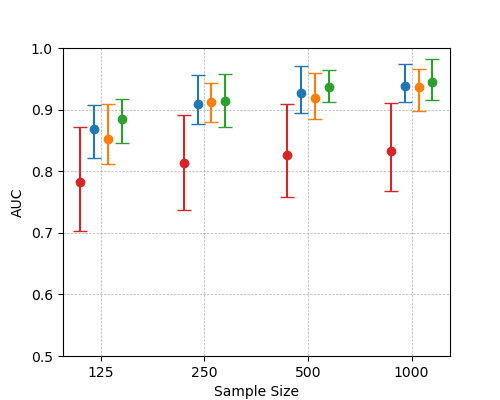}
        \subcaption{R-vine-Positive}\label{fig:rvine_posonly}
    \end{minipage}
    \begin{minipage}{0.32\textwidth}
        \centering
        \includegraphics[width=\textwidth]{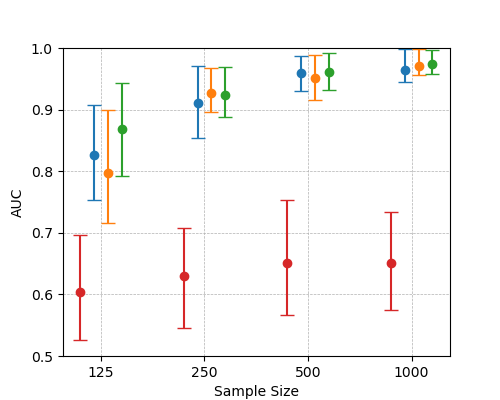}
        \subcaption{C-vine}
    \end{minipage}
    \caption{Across different vine scenarios, the figure depicts the $\auc$ of the GTM in identifying full conditional independencies based on the $\iae$ between the GTM and the GTM with independence assumption for each dimension pair. The GTM performance for different training sample sizes is presented for the model without a LASSO penalty in blue, with a LASSO penalty in orange and with an adaptive LASSO penalty in green and the benchmark GGM in red. The $\iae$ scores are approximated via $10,000$ synthetic samples from the GTM. The dots represent the means, and the whiskers the $20\%$ and $80\%$ quantiles across $30$ simulation replications.}
    \label{fig:xvine_auc_iae}
\end{figure}

\subsection{Main Findings from Simulations}
We conclude the simulation study by recapping the main findings.
First, the GTM performs effectively in both learning the true underlying distribution and identifying conditional independencies.
Notably, it often surpasses GGMs in both areas across scenarios, achieving strong performance down to $250$ observations for standard scenarios and $500$ observations for the R-vine-weak scenario.
The GTM particularly excels in identifying conditional independencies when dependencies are more nonlinear, indicated by larger Kendall's $\tau$ in the pair copulas, and when independencies are challenging to identify due to shorter conditional dependency paths between variables, as found in C-vine structures.
Second, we find that an adaptive LASSO penalty improves both the learning of the underlying distribution as well as the identification of independencies for smaller sample sizes down to $125$. 
Therefore, for practical purposes, we suggest training a simple GTM, followed by an adaptive-LASSO GTM based on the first GTM's local pseudo-precision matrix.

\section{Complex Distributions: MAGIC Application} \label{sec:complex_distribution}

To further evaluate the performance of our model, particularly in terms of learning complex real-world multivariate distributions beyond the capabilities of classical parametric copulas, we compare our model against the non-parametric vine copula approach detailed by \cite{nonpar_copula_evade_cod}. The authors demonstrate the effectiveness of constructing a vine from non-parametric pair copulas using a dataset that simulates measurements from the MAGIC (Major Atmospheric Gamma-ray Imaging Cherenkov) telescopes located in the Canary Islands. 
This dataset, publicly available on the University of California Irvine (UCI) Machine Learning repository website\footnote{See here: \url{https://archive.ics.uci.edu/dataset/159/magic+gamma+telescope}}, features complex multivariate relationships, making it well-suited for evaluating our model. Originally employed in a case study \citep{magic_case_study}, the objective was to differentiate gamma rays (signal) from hadron showers (noise) observations. The dataset consists of $19,020$ observations across $10$ continuous dimensions that describe the shape, size, orientation, intensity, and asymmetry of the Cherenkov images, with $12,332$ observations classified as gamma rays and $6,688$ classified as hadron showers.
The non-parametric vine copula \citep{nonpar_copula_evade_cod} significantly outperformed the classification methods applied in \cite{magic_case_study} by training separate models for each class and subsequently utilizing the computed class probabilities within a Bayesian classifier. 
This result is particularly notable, given that the models were not specifically designed for classification; instead, they accurately learned the underlying distributions to the extent that they could differentiate observation classes by comparing likelihoods for each group.

We replicated the training procedure of \cite{nonpar_copula_evade_cod} and utilized the likelihoods obtained from the GTMs within a Bayes classifier. The dataset was divided such that $2/3$ comprised the training sample and $1/3$ served as the control sample \citep{magic_case_study}. For the GTM architecture, several hyperparameters needed to be selected, including the number of knots in the transformation layer splines, the number of knots in the decorrelation layer spline, and the number of decorrelation layers. Additionally, it was necessary to identify the optimal penalties: the second-order ridge penalty for the marginal transformation P-splines ($\tau_4$), the first ($\tau_1$) and second-order ridge penalties ($\tau_2$) for the decorrelation layer P-splines, and a potentially weighted LASSO penalty ($\tau_3$).
Regarding the number of knots in the transformation layer, we trained separate unpenalized CTM on each marginal, progressively increasing the number of knots until the one-dimensional latent space attained a closeness to normality, as determined by a Shapiro-Wilk test p-value of $0.01$ or higher.
In the decorrelation layer, we implemented $30$ knots without further tuning. For the number of decorrelation layers, we trained the model across depths of $L \in \{3, 4, \ldots, 9\}$, performing $40$ hyperparameter draws for the penalties for each depth. We divided the training sample into $80/20$ for training and validation, thus avoiding additional cross-validation. To determine the optimal model depth, penalties, and early stopping criteria, we selected the models with the highest likelihood on the validation set for each group respectively. We found that the optimal number of decorrelation layers was $L_h=8$ and $L_g=6$.
We did not observe significant improvements by applying a weighted LASSO penalty ($\tau_3$), likely due to the large size of the training dataset. For the Gauss-Legendre Quadratures required by \autoref{alg:llr_metrics}, we used $20$ quadrature points spanning $[-15,15]$, computing them in the space of $\bftz$ after the marginal transformation.
In \autoref{tab:magic_fpr_tpr}, we present the true positive rate (TPR) of the classification across different values for the false positive rate (FPR), as reported by \cite{nonpar_copula_evade_cod}, alongside results obtained with the GTM.
The results demonstrate that the GTM improves upon both benchmarks at false positive rates of 0.02, 0.05, 0.1 and 0.2.
At a FPR of 0.01, it is slightly below both competitors.
Overall, we conclude that the GTM outperforms the multivariate Kernel Density Estimator (MVKDE) and is competitive to the vine copula (Vine).
A minimal $\text{GTM}_3$ with $L=3$ decorrelation layers, also reported in \autoref{tab:magic_fpr_tpr}, achieved TPRs that outperform the Kernel Density estimator in all five FPRs and is comparable to the vine copula and the larger GTM, even slightly better at smaller FPRs, indicating that even a small $\text{GTM}_3$ is competitive.
\begin{table}[htbp]
    \centering
    \begin{tabular}{lrrrrr}
        \toprule
        FPR & 0.01 & 0.02 & 0.05 & 0.10 & 0.20 \\
        \midrule
        Vine              & 0.335 & 0.428 & 0.652 & 0.780 & 0.918 \\
        MVKDE             & 0.335 & 0.408 & 0.567 & 0.730 & 0.868 \\
        GTM               & 0.323 & 0.451 & 0.663 & 0.815 & 0.927 \\
        $\text{GTM}_3$    & 0.342 & 0.453 & 0.639 & 0.796 & 0.916 \\
        \bottomrule
    \end{tabular}
    \caption{Comparison of true positive rates (TPR) for a given target false positive rate (FPR) for different models.
    We compare The GTM to the benchmark vine copula (vine) and Multivariate Kernel Density Estimator (mvkde) from \cite{nonpar_copula_evade_cod}.
    }
    \label{tab:magic_fpr_tpr}
\end{table}
To illustrate key insights into the GTM, we present results using data from the hadron group; results for the gamma group are comparable. We focus on the hadron data primarily because it comprises fewer observations, with $3,568$ samples in the training data, thereby better showcasing the model's capability. 
To demonstrate the GTM's approximation of the complex data distribution, we select some particularly intriguing pair plots from the training data observations and juxtapose them with an identical number of synthetically sampled data points from the model in \autoref{fig:magic_h_pairplots_train_synthetic}.
The results indicate that the GTM effectively approximates the varying complex patterns in the data distribution. However, the GTM also tends to smooth over the data, missing sharp details in low-density regions, such as the negative relationship at the bottom right of subplot (e) or the variance extent in the outskirts of the x-shaped density in subplot (f).

\begin{figure}[htbp]
  \centering
  {\footnotesize Training samples} \\
  \begin{subfigure}[b]{0.15\linewidth}
    \includegraphics[width=\linewidth]{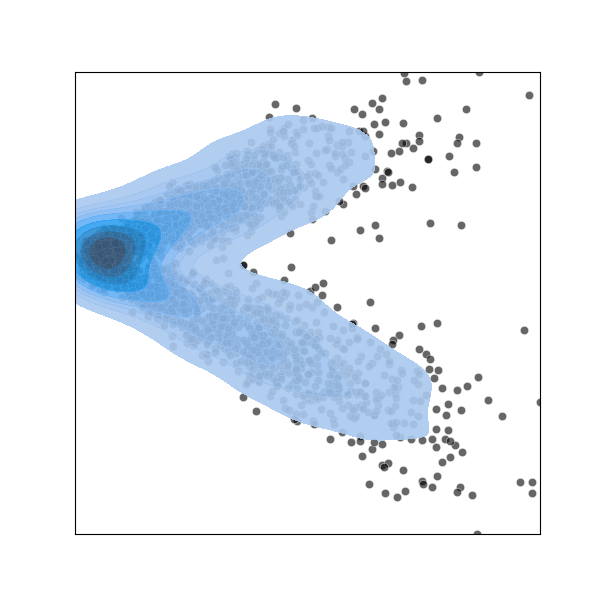}
    \caption{$Y_0,Y_6$}
  \end{subfigure}
  \begin{subfigure}[b]{0.15\linewidth}
    \includegraphics[width=\linewidth]{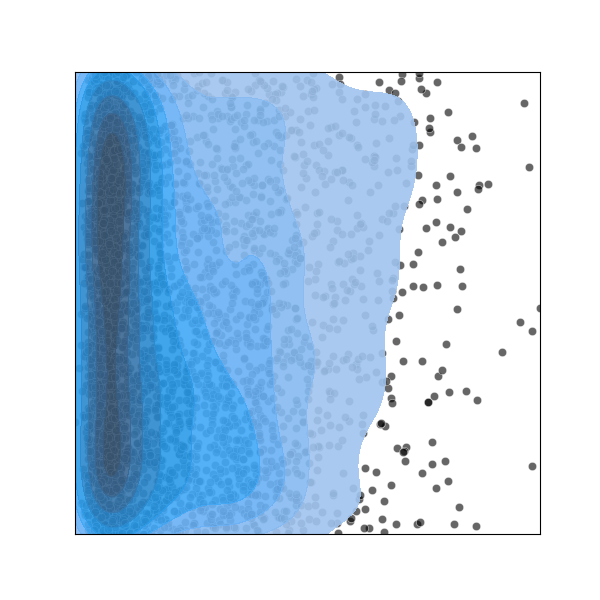}
    \caption{$Y_0,Y_8$}
  \end{subfigure}
  \begin{subfigure}[b]{0.15\linewidth}
    \includegraphics[width=\linewidth]{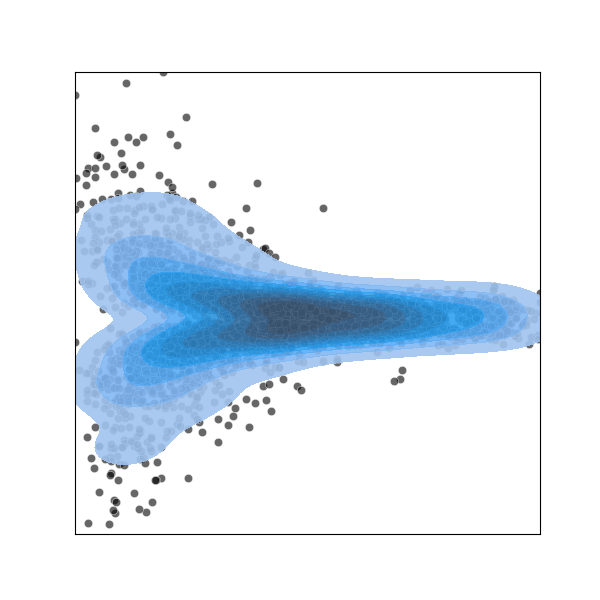}
    \caption{$Y_3,Y_7$}
  \end{subfigure}
  \begin{subfigure}[b]{0.15\linewidth}
    \includegraphics[width=\linewidth]{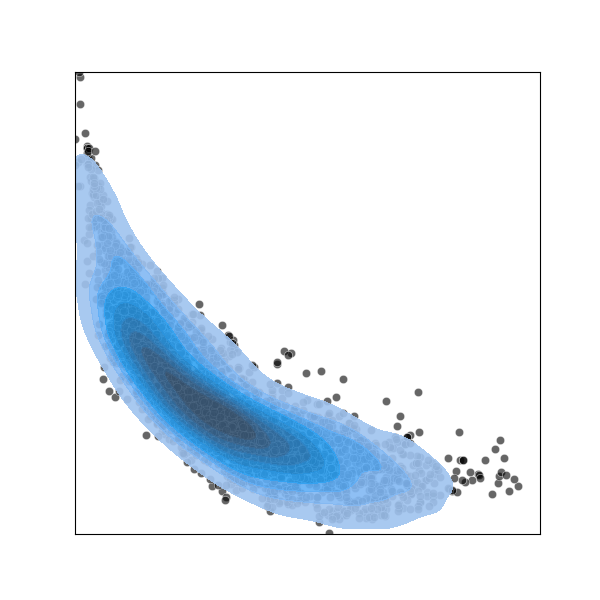}
    \caption{$Y_4,Y_2$}
  \end{subfigure}
  \begin{subfigure}[b]{0.15\linewidth}
    \includegraphics[width=\linewidth]{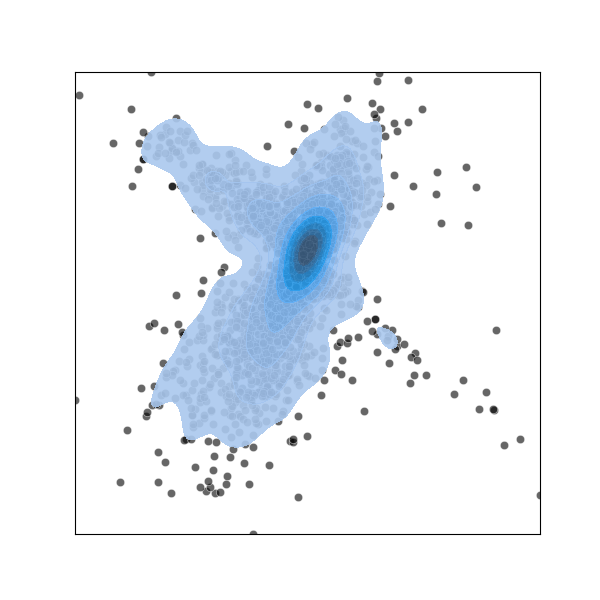}
    \caption{$Y_5,Y_6$}
  \end{subfigure}
  \begin{subfigure}[b]{0.15\linewidth}
    \includegraphics[width=\linewidth]{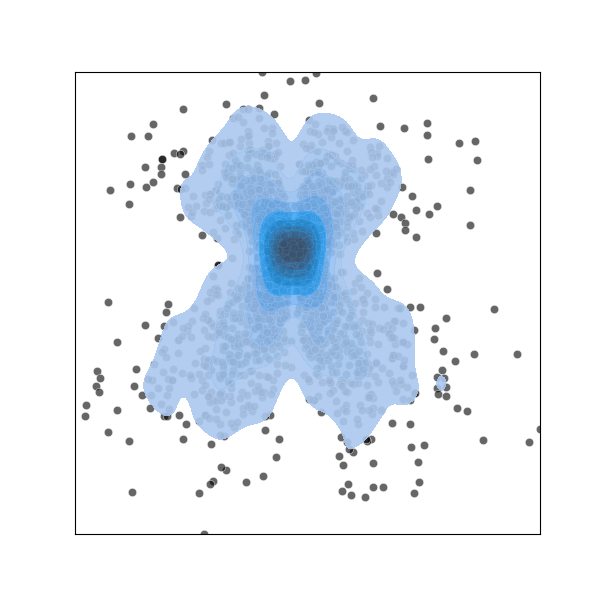}
    \caption{$Y_7,Y_6$}
  \end{subfigure}

  {\footnotesize Synthetic samples from the model} \\
  \begin{subfigure}[b]{0.15\linewidth}
    \includegraphics[width=\linewidth]{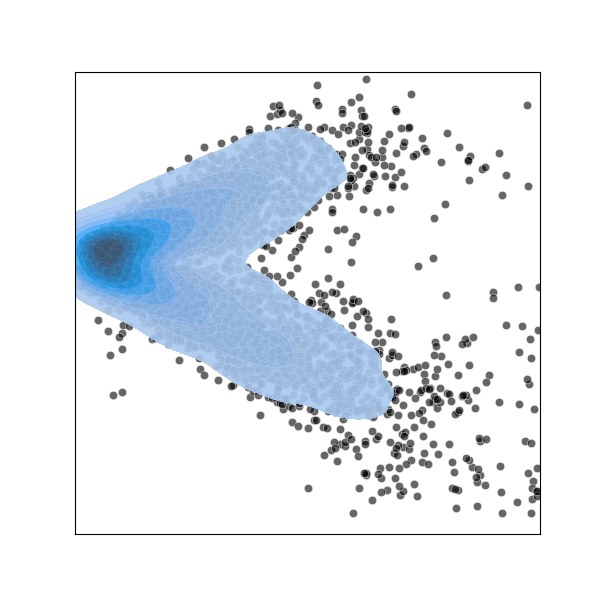}
  \end{subfigure}
  \begin{subfigure}[b]{0.15\linewidth}
    \includegraphics[width=\linewidth]{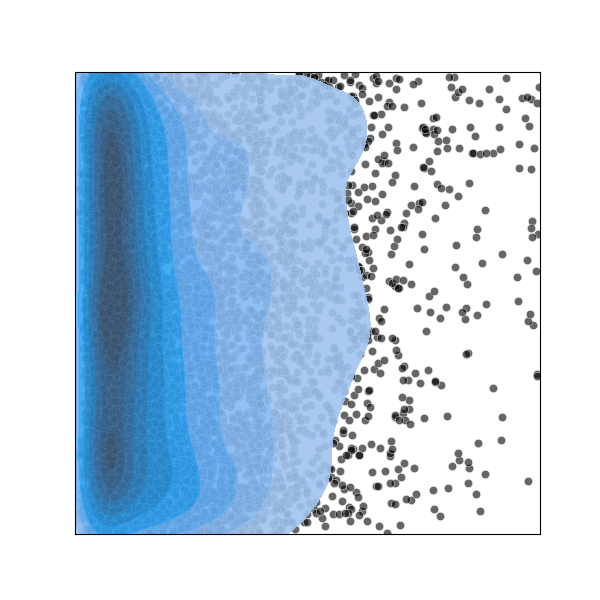}
  \end{subfigure}
  \begin{subfigure}[b]{0.15\linewidth}
    \includegraphics[width=\linewidth]{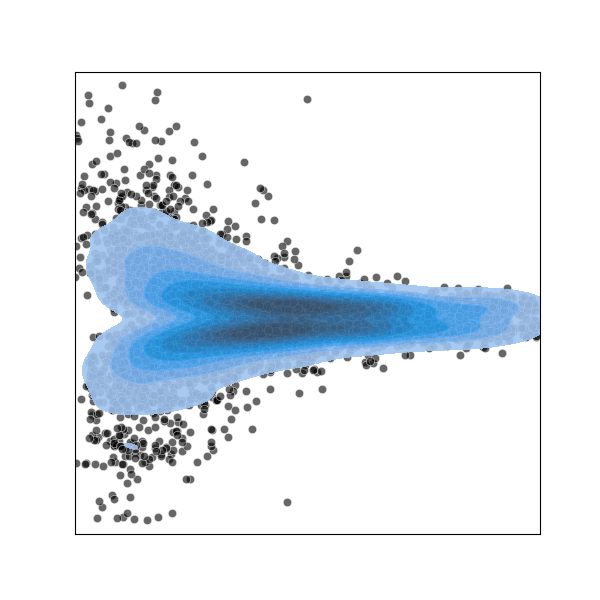}
  \end{subfigure}
  \begin{subfigure}[b]{0.15\linewidth}
    \includegraphics[width=\linewidth]{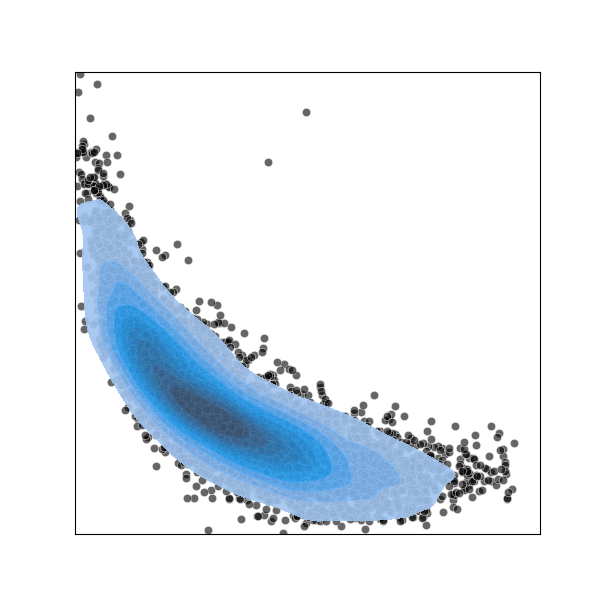}
  \end{subfigure}
  \begin{subfigure}[b]{0.15\linewidth}
    \includegraphics[width=\linewidth]{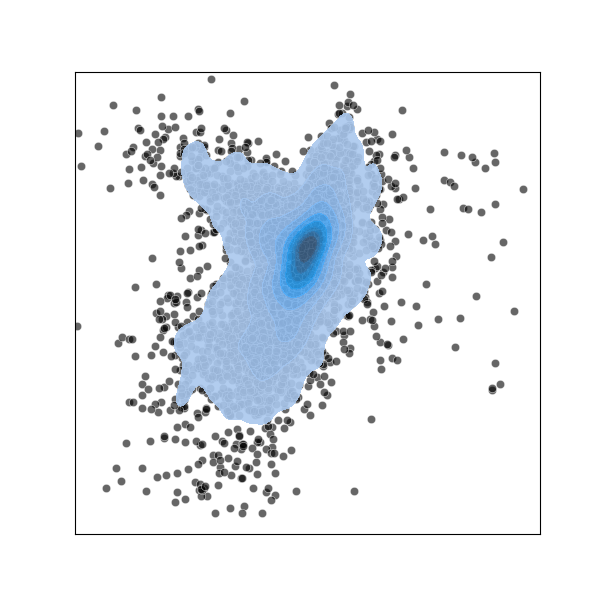}
  \end{subfigure}
  \begin{subfigure}[b]{0.15\linewidth}
    \includegraphics[width=\linewidth]{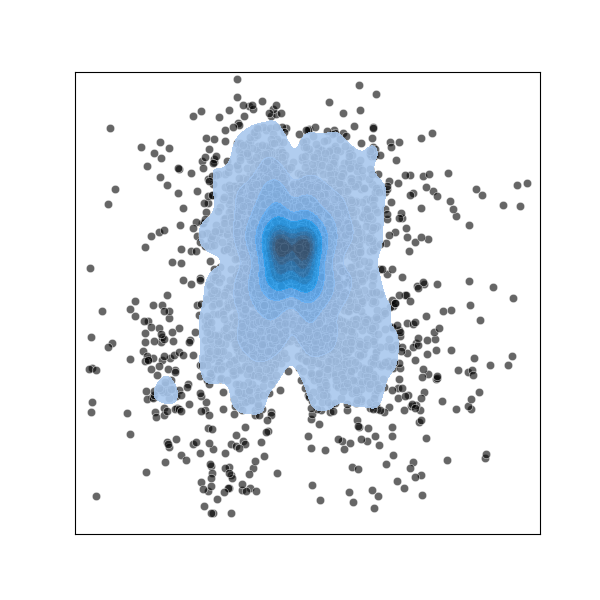}
  \end{subfigure}
  
  \caption{Subset of pairplots from the hadron class, where the captions state the dimensions depicted. The first row are the $3568$ training set samples. The second row are $10000$ synthetically sampled observations from the model. 
  }
  \label{fig:magic_h_pairplots_train_synthetic}
\end{figure}

One particularly attractive feature of the GTM is its ability to interpret patterns in the learned distribution through the lens of local conditional pseudo-correlation patterns. \autoref{fig:graph_magic_data_h} illustrates the learned conditional dependency graph for the hadron data, displaying only pairs with a conditional dependence $\iae$ above $0.1$.

\begin{figure}[htbp]
  \centering
  \begin{subfigure}[b]{1\linewidth}
    \includegraphics[width=\linewidth]{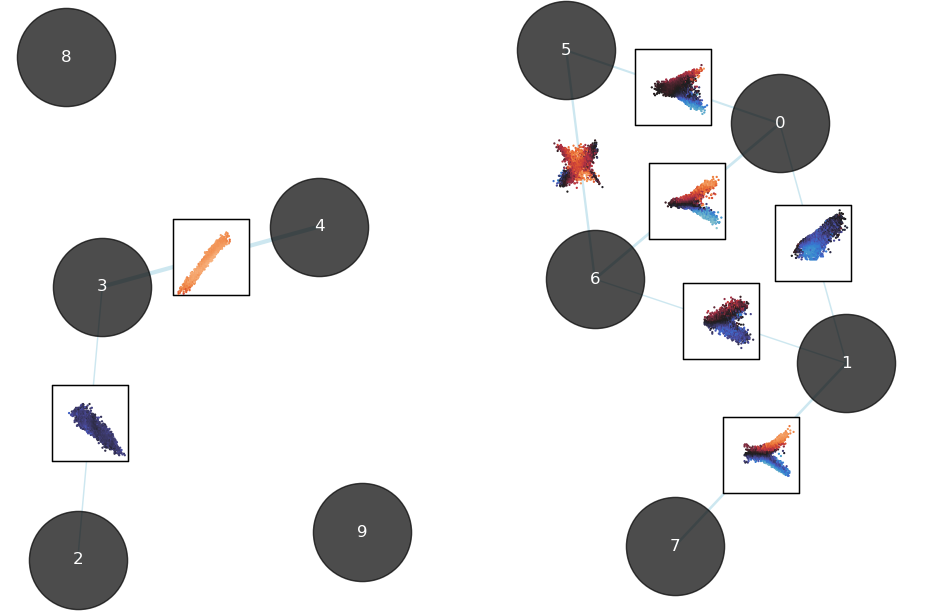}
  \end{subfigure}
  \caption{The graph depicts the full conditional dependency Graph for the Hadron dataset. Only nodes with a conditional dependence $\iae$ above $0.1$ are linked with edges. The nonlinear dependencies after marginal transformations are overlayed on each edge. }
  \label{fig:graph_magic_data_h}
\end{figure}

In other words, we focus on pairs where the learned conditional dependency results in at least a $10\%$ non-overlap between the full model-implied conditional density and conditional independence density on average across all conditioning sets.
Given the nonlinear nature of dependencies, we overlay the pair plots to demonstrate the type of dependency.
\autoref{fig:sign_conditional_corr_h} further showcases these eight conditional dependencies in both the original space of $\bfy$ on the left (a) and the space of $\tilde{\mathbf{z}}$ after transforming the marginals on the right (b) for the training data. We note three major observations: First, dependencies become closer to linear after marginal transformation in three instances, yet remain clearly nonlinear in the other five cases. Second, local conditional pseudo-correlations can be interpreted as positive where there seems to be a positive local relationship and negative where there is a negative local relationship. Finally, the local conditional pseudo-correlations are strongly linked to conditional independence, evident from their tendency to shrink and approach zero as the $\iae$ decreases.
Indeed, this observation is also reflected in a Spearman and Pearson correlations between the mean absolute local conditional pseudo-correlations and the $\iae$ for the $45$ pairs of $0.83$ and $0.94$ in the hadron and $0.89$ and $0.96$ in the gamma dataset.
In appendix \autoref{fig:unsign_conditional_corr_h}, we provide pairwise local pseudo-conditional correlation plots for the remaining non-significant pairs.

\begin{figure}[htbp]
  \centering
  \begin{subfigure}[b]{0.45\linewidth}
  \caption{$Y$}
    \includegraphics[width=\linewidth]{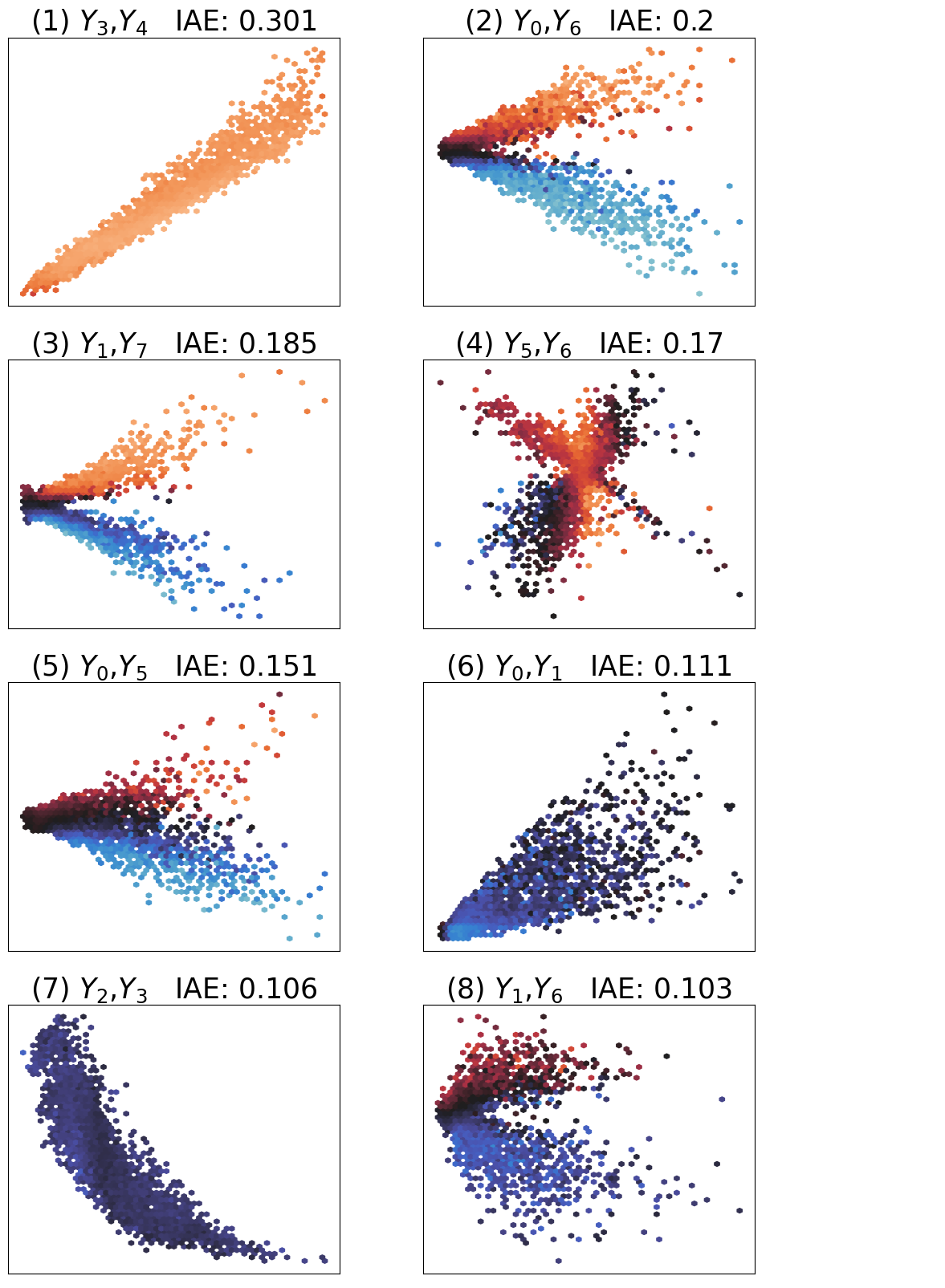}
  \end{subfigure}
  \begin{subfigure}[b]{0.45\linewidth}
  \caption{$\tilde{Z}$}
    \includegraphics[width=\linewidth]{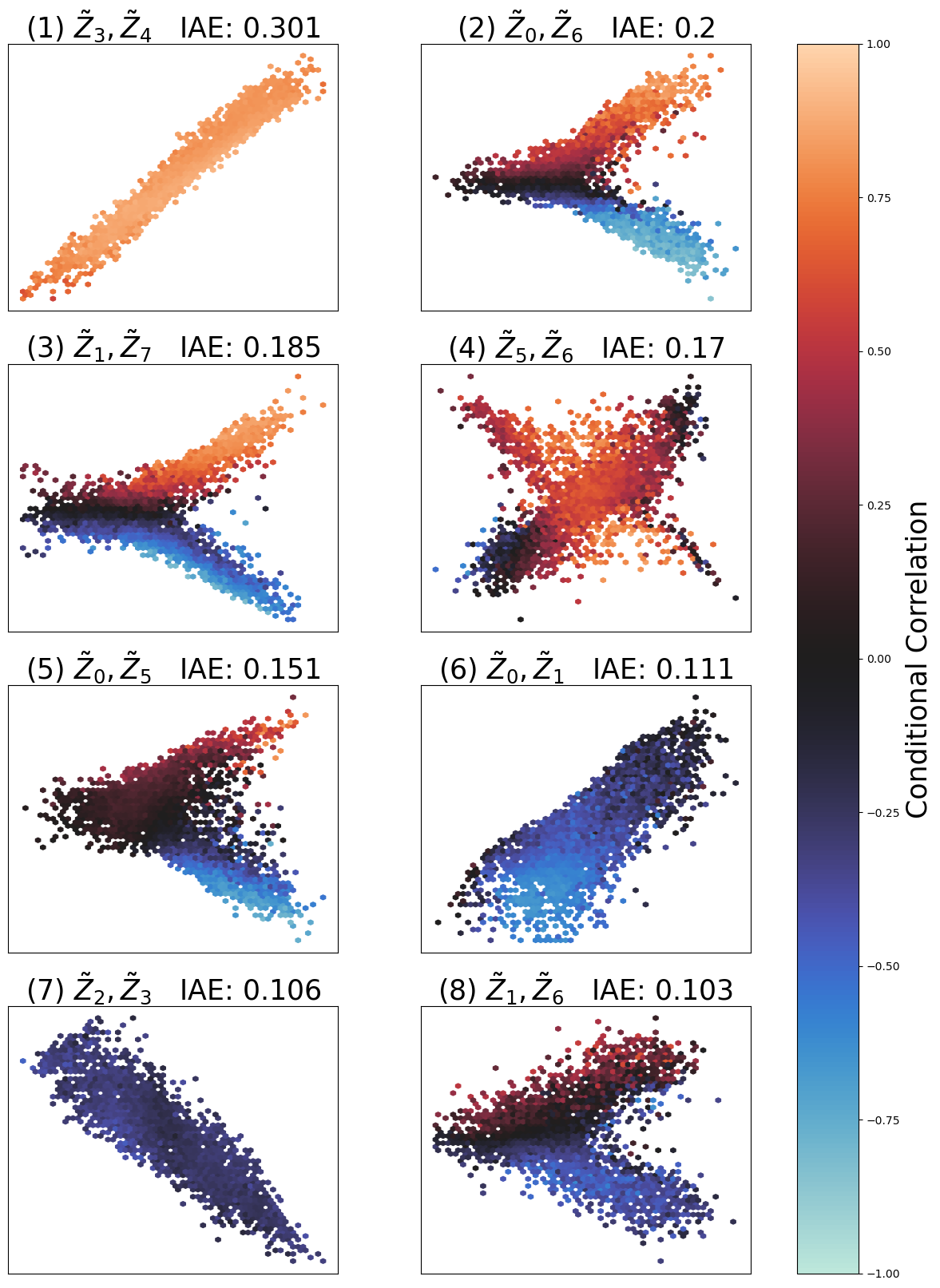}
  \end{subfigure}
  \caption{The eight conditionally dependent pairs in the hadron dataset ordered by their $\iae$ conditional dependence metrics.
  The plots on the left depict the original training data and the ones on the right depict the data after the marginal transformation.
  Both are colored by the local conditional pseudo-correlations.}
  \label{fig:sign_conditional_corr_h}
\end{figure}

The interpretation appears to fail in plots (4) and (6) of \autoref{fig:sign_conditional_corr_h}. In subplot (4), there are instances where a negative relationship shows positive local conditional pseudo-correlations, and even more clearly in subplot (6), a distinctly positive linear relationship exhibits negative local conditional pseudo-correlations. These initially puzzling results can be attributed to the fact that both sets of variables are not only directly but also indirectly connected through other variables, as depicted in \autoref{fig:graph_magic_data_h}. Specifically, $Y_0$ and $Y_1$ are indirectly linked via $Y_6$, while $Y_5$ and $Y_6$ are indirectly connected through $Y_0$. 
To better interpret these relationships, it is prudent to examine the conditional dependence of these pairs more closely. We illustrate this by selecting two example training data points from each pair and generating conditional samples by keeping the other eight conditioning variables fixed. Using importance sampling, we plot the resulting synthetic samples in the rows of \autoref{fig:conditional_correlation_sampling}. The importance sampling procedure is described in Appendix \autoref{alg:cond_samples}. The first two plot columns display the sampled data in the $\bfy$ space, initially in terms of the bivariate conditional density and subsequently with overlaid local conditional pseudo-correlations. The third and fourth plot columns similarly depict the density and local conditional pseudo-correlations after marginal transformation to $\tilde{\mathbf{z}}$. 

The latter are critical because actual local conditional pseudo-correlations are learned in the transformed space. All four plot columns now present clearer patterns: positive or negative local conditional pseudo-correlations correlate with positive or negative relationships. This observation holds consistently across both the MAGIC data and the simulation study. 
\begin{figure}[htbp]
  \centering
  \begin{subfigure}[b]{1\linewidth}
  \caption{$Y_0, Y_1$}
    \includegraphics[width=\linewidth]{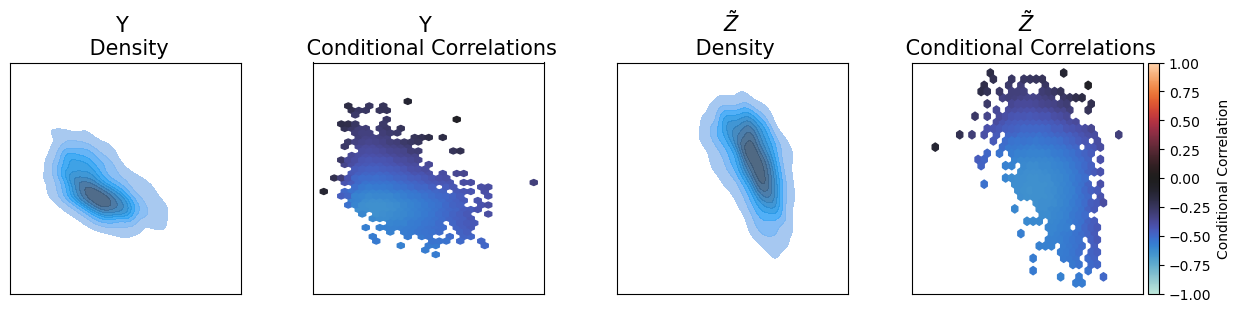}
  \end{subfigure}
  \begin{subfigure}[b]{1\linewidth}
  \caption{$Y_5, Y_6$}
    \includegraphics[width=\linewidth]{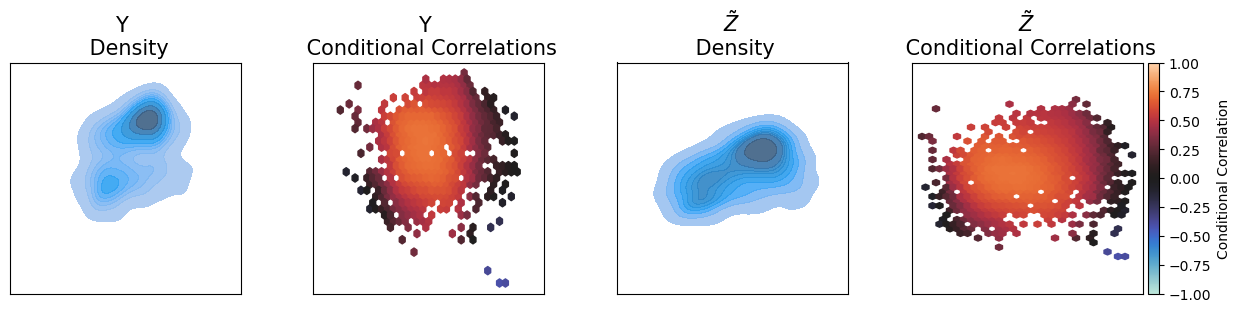}
  \end{subfigure}
  \caption{Conditional samples of pairs $Y_0,Y_1$ and $Y_5,Y_6$ at one illustrative training observation point.
  The samples are created by sampling the respective dimension while keeping all other dimensions fixed at the training observation point.
  The first subplot depicts the density of the samples in the data space $Y$.
  The second also depicts the samples in the data space $Y$ and is colored by the local conditional pseudo-correlation pattern.
  The last two plots analogously plot the density and the local conditional pseudo-correlations for the marginally transformed samples $\tilde{Z}$.}
  \label{fig:conditional_correlation_sampling}
\end{figure}

\section{Conclusion} \label{sec:conclusion}

In this work, we introduced Graphical Transformation Models (GTMs)\footnote{The GTM python code is available on GitHub at \url{https://github.com/MatthiasHerp/gtm}} as a novel approach for modeling multivariate distributions with complex dependencies, while preserving the interpretability of conditional independence structures. By extending Multivariate Conditional Transformation Models (MCTMs) with a flexible sequence of transformations inspired by normalizing flows, we demonstrated how GTMs bridge the gap between Gaussian graphical models and highly flexible but opaque machine learning methods. In achieving this, we incorporated spline penalties that balance between a Gaussian copula and a fully nonparametric normalizing flow, providing a reasonable baseline, controlling flexibility and penalize towards conditional independencies. Additionally, we offered an interpretation of the conditional dependency structure through local conditional pseudo-correlations.

Certain limitations and avenues for future research remain. While GTMs offer an interpretation of dependencies via local conditional pseudo-correlations, these are not truly metrics of conditional independence. Therefore, our adaptive LASSO serves as only an approximate conditional independence penalty. To address this, we introduced the $\iae$ and $\kld$ metrics for a more rigorous quantification of conditional independence, and future research will aim to develop a likelihood-ratio-based test for formal statistical inference. 
Additionally, GTMs are not true copulas because the marginals after the transformation layer are not uniformly distributed. This limitation may be mitigated by training the marginals separately from the joint structure \citep[as in][]{wiese2019copulamarginalflows}, resulting in a latent space with approximately standard Gaussian marginals. These marginals can then be transformed using the Gaussian cumulative distribution function to define the sequence of decorrelation layers effectively as a copula.
Furthermore, an important direction for future work is the integration of covariates to enhance the flexibility and applicability of GTMs. An initial step could involve the inclusion of interpretable marginal location and scale effects as proposed by \cite{brachemBayesianPenalizedTransformation2024}.

Overall, GTMs offer a flexible and interpretable framework for modeling multivariate dependencies. They facilitate moving beyond a Gaussian Copula in a nonparametric manner, while still enabling the identification, interpretation, and approximate penalization of full conditional dependencies. We foresee that this approach will prove beneficial for creating nonparametric undirected graphs across various applications, particularly in fields where understanding the structure of conditional dependencies is crucial and the assumption of linear or even monotonic effects is too restrictive.

\bibliographystyle{jasa3}
\bibliography{bib.bib}

\appendix

\section{Notation} \label{app:notation}
\begin{table}[!htbp]
    \centering
    \begin{tabular}{p{7.2cm}p{7.2cm}}
        \toprule
        \textbf{Symbol} & \textbf{Description} \\
        \midrule
        $Y = (Y_1, Y_2, ..., Y_J)^T \in \mathbb{R}^J$ & Response random variable vector of dimension $J$ \\
        $\mathbf{y}_i = (y_{i,1}, y_{i,2}, ..., y_{i,J})^T \in \mathbb{R}^J$ & Observation $i$ of response variable vector \\
        $y_{i,j} \in \mathbb{R}$ & Observation $i$ of response variable $j$ \\
        $\begin{bmatrix} \mathbf{y}^1 & \mathbf{y}^2 & \dots & \mathbf{y}^S \end{bmatrix}^T \in \mathbb{R}^{N \times J}$ & Observation Sample of Size $S$ for the response vector \\
        \midrule
        $f(Y)$ & Probability density function (pdf) \\
        $F(Y)$ & Cumulative distribution function (cdf) \\
        \midrule
        $\mathbf{\tilde{h}}(Y) = 
        (\tilde{h}_1(Y_1), \tilde{h}_2(y_2), ..., \tilde{h}_J(Y_J))^T$ & Marginal transformation for each response dimension $1,...,J$ \\
        $Z = \mathbf{h}(Y) = 
        (Z_1, Z_2, ..., Z_J)^T
        $ & Completely transformed response random variable vector to latent space of dimensionality $J$ \\
        \bottomrule
        $\mathbf{a}(y_{i,j}) = (a_1(y_{i,j}), a_2(y_{i,j}), ..., a_K(y_{i,j}) )^T$ & basis function vector of length $K$ \\
        $\mathbf{\vartheta} = (\vartheta_1, \vartheta_2, ..., \vartheta_K) )^T$ & basis parameter vector of length $K$ \\
        \bottomrule
        $\mathbf{\Lambda}(\tilde{Z}_{l-1})_l \in \mathbb{R}^{J \times J}$ & Lambda Matrix of layer $l$, for Layers $1,2,...,L$\\
        $\tilde{Z}_0 = \tilde{Z} = \mathbf{\tilde{h}}(Y)
        $ & Marginally transformed response random variable vector of dimensionality $J$ \\
        $\tilde{Z}_l = \mathbf{\Lambda}(\tilde{Z}_{l-1})_l  \tilde{Z}_{l-1} = 
        (\tilde{Z}_{1,l}, \tilde{Z}_{2,l}, ..., \tilde{Z}_{J,l})^T
        $ & Latent Space vector of dimensionality $J$ after Layer $l$ \\
        $\lambda_{r,c,i}(\mathbf{\tilde{z}}_{i,l-1}) \in \mathbb{R}$ & Lambda Matrix entry for row $r$, column $c$ of observation $i$ in layer $l$ \\
        $\mathbf{P}(\tilde{Z}_0) \in \mathbb{R}^{J \times J}$ & Local Pseudo Precision Matrix \\
        $p(\mathbf{\tilde{z}}_i)_{r,c} \in \mathbb{R}$ & Local Pseudo Precision Matrix entry for row $r$, column $c$ of observation $i$ \\
        $\rho(\mathbf{\tilde{z}}_i)_{r,c} \in \mathbb{R}$ & Local Pseudo Conditional Correlation for row $r$, column $c$ of observation $i$ \\
        \bottomrule
        $ \tau $ & Penalty hyperparameter \\
        $ g() $ & Some invertible function \\
        $ c() $ & conditioner function of a normalizing flow \\
        $ \ell_i(\mathbf{y}_i) $ & Log-Likelihood contribution of response vector observation $i$ \\
        $\mathbf{F} \in {\{0,1\}}^{J \times J}$ & Flipping matrix which flips a vector\\
        
    \end{tabular}
    \caption{Summary of notation used in the paper.}
    \label{tab:notation}
\end{table}

\FloatBarrier

\section{Challenge of Identifying Independencies} \label{appendix:challenge_independence_identification}

As discussed in \autoref{theory:mctm_ggm}, zero entries in the precision matrix $\mathbf{P}$ signify conditional independence in the MCTM. However, the sequential nonlinear layers $\mathbf{\Lambda}$ of the GTM, which enable fitting complex data structures, also make identifying conditional independencies more challenging. We illustrate this issue using an example of a GTM with three layers for three-dimensional data, i.e., $J=L=3$.
The transformation layer does not induce dependence, so we will focus on the decorrelation layers and assume $\bfth(\bfy)=\bfy = \bftz$ for simplification. We aim to highlight two main points: First, that every element of the local pseudo-precision matrix $\mathbf{P}(\bfy)$ is dependent on every input $\bfy$. 
Second, that this dependence results in zero off-diagonal element not necessarily implying conditional independence — that is, $p_{u,v}=0$ does not mean $y_u \independence y_v | y_{\backslash {u,v} }$. In our example, $p_{3,1}=0$ does not imply $y_3 \independence y_1 | y_{2}$.

To better illustrate the first point, we substitute the functions $\lambda_{r,c,l}(\bftz_l)$ with variables $[a, b, ..., g]$ in each layer matrix $\mathbf{\Lambda}_l(\bfy)$ and define the joint lambda matrix as:
\begin{align*}
\mathbf{\Lambda}(\bfy) = \mathbf{\Lambda}_3(\bftz_2) \mathbf{\Lambda}_2^T(\bftz_1) \mathbf{\Lambda}_1(\bfy) =
\begin{pmatrix}
1 & 0 & 0 \\
g & 1 & 0 \\
h & i & 1
\end{pmatrix}
\begin{pmatrix}
1 & 0 & 0 \\
d & 1 & 0 \\
e & f & 1
\end{pmatrix}^T
\begin{pmatrix}
1 & 0 & 0 \\
a & 1 & 0 \\
b & c & 1
\end{pmatrix}
= 
\end{align*}
\begin{align*}
\begin{pmatrix}
1 + ad + be & d + ce & e \\
g + a(1 + dg) + b(f + eg) & 1 + dg + c(f + eg) & f + eg \\
h + a(dh + i) + b(1 + eh + fi) & dh + i + c(1 + eh + fi) & 1 + eh + fi \\
\end{pmatrix}
\end{align*}
The local pseudo-precision matrix is defined as $P = \mathbf{\Lambda}(\bfy)^{\sfT} \mathbf{\Lambda}(\bfy)$. Suppose we assume that $y_1$ and $y_3$ are conditionally independent given $y_2$. We would then focus on the splines that influence the entry of the precision matrix in the third row, first column:
\begin{align*}
p_{3,1} = 
\begin{pmatrix}
e & f +eg & 1 + eh + fi
\end{pmatrix}
\begin{pmatrix}
1 + ad + be \\
g + a(1 + dg) + b(f + eg) \\
h + a(dh + i) + b(1 + eh + fi)
\end{pmatrix}
\end{align*}
If we then set all three splines connecting the two variables, namely $b,e,h$ to be zero as there is no interaction we are left with:
\begin{align*}
p_{3,1} = 
\begin{pmatrix}
0 & f +0g & 1 + 00 + fi
\end{pmatrix}
\begin{pmatrix}
1 + ad + 00 \\
g + a(1 + dg) + 0(f + 0g) \\
0 + a(d0 + i) + 0(1 + eh + fi)
\end{pmatrix}
=
\end{align*}
\begin{align*}
\begin{pmatrix}
0 & f & 1 + fi
\end{pmatrix}
\begin{pmatrix}
1 + ad \\
g + a(1 + dg) \\
ai
\end{pmatrix} 
\end{align*}
which is equal to: 
\begin{align*}
p_{3,1} = f * (g + a(1 + dg)) + (1 + fi) * ai
\end{align*}
This illustrates the type of restrictions that impact the splines modeling non-independent relationships—specifically those between $y_1$ and $y_2$, involving splines $a$, $d$, $g$, and those between $y_2$ and $y_3$, involving splines $c$, $f$, $i$. Intuitively, this occurs because the effect of $y_3$ on $y_2$ can influence $y_1$ in the subsequent layer via $y_2$’s effect on $y_1$. The reverse is also true. Similar computations reveal that all other elements $p_{1,1}, p_{2,2}, p_{3,3}, p_{2,1}, p_{3,2}$ depend on all input dimensions $y_1, y_2, y_3$.

Knowing that all elements of $\bfP(\bfy)$ depend on $\bfy$, we express the density $f(\bfy)$ in the Gaussian latent space of $\mathbf{\Lambda}_l(\bfy) \bfy$, take the logarithm, and drop the constant term to obtain:
\begin{align*}
\log f(\mathbf{\Lambda}_l(\bfy) \bfy) = \log(f(y_1, y_2, y_3)) \propto 
\end{align*}
\begin{align*}
p_{1,1}(\bfy) y_1^2 + 
p_{2,2}(\bfy) y_2^2 + 
p_{3,3}(\bfy) y_3^2 +
p_{2,1}(\bfy) y_2 y_1 + 
p_{3,1}(\bfy) y_3 y_1 + 
p_{3,2}(\bfy) y_3 y_2 
\end{align*}
From this, we observe that each entry in the local pseudo precision matrix may depend on all inputs. Therefore, even if $p_{3,1}(\bfy)=0$, it does not lead to conditional independence because $y_1$ and $y_3$ still appear in the same sum element. In other words, in density—not log space—$y_1$ and $y_3$ do not necessarily factorize, indicating they are not necessarily conditionally independent. This is because additional requirements would involve each other entry $p_{u,v}(\bfy)$, where $\{u,v\} \in \{\{1,1\}, \{2,2\}, \{3,3\}, \{2,1\}, \{3,1\}, \{3,2\}\}$, depending solely on either $y_1$ or $y_3$. Hence, the restriction $p_{3,1}(\bfy)=0$ is insufficient.
Nevertheless, as shown in our simulation study \autoref{sec:simstudy_vine_copulas} and application \autoref{sec:complex_distribution}, $p_{3,1}(\bfy)=0$ serves as a good approximate indicator for conditional independence in the GTM, even for deeper models.

\section{Computational Details} \label{app:comp} 

\subsection{Implementation} \label{app:implementation}
The model was implemented using the PyTorch framework, selected for its flexible auto-differentiation capabilities \cite{pytorch}.

B-spline evaluations were computed utilizing De Boor's algorithm \cite{deBoor_bsplines}, which optimizes performance by evaluating only non-zero bases rather than all basis functions and taking the weighted sum. 

For synthetic sampling, the inverse of both decorrelation and transformation layers must be computed. 
The decorrelation layers come with a straightforward closed-form inverse, as typical of any coupling layer, described in \autoref{alg:inverse_decorrelation}.

\begin{algorithm}[H]
\caption{Calculating the Inverse Decorrelation Layer}
\label{alg:inverse_decorrelation}
\textbf{Input:} The output of $\mathbf{\Lambda}_l(\bftz_{l-1})$ Layer $l$ resulting in $\bftz_{l}$ \\
\textbf{Output:} The input to the $\mathbf{\Lambda}_l(\bftz_{l-1})$ Layer $l$, i.e., $\bftz_{l-1}$
\begin{algorithmic}
\FORALL{Data Dimensions $j \in [1,2,...,J]$ iteratively}
    \STATE \[
    \tilde{z}_{j,l-1} = \tilde{z}_{j,l} -  \sum_{i=1}^{j-1} \lambda_{i,l-1}(\tilde{z}_{i,l-1})\tilde{z}_{i,l-1}
    \]
\ENDFOR
\end{algorithmic}
\end{algorithm}

For the transformation layer, although it is invertible, it lacks a closed-form solution for the inverse transformation. However, we can achieve a highly precise inverse with a workaround, as detailed in \autoref{alg:inverse_transformation}.

\begin{algorithm}[H]
\caption{Approximating the Inverse Transformation Layer}
\label{alg:inverse_transformation}
\textbf{Input:} Minimum and maximum values of the input $\mathbf{y}$ \\
\textbf{Output:} Approximation of the inverse transformation layer
\begin{algorithmic}
    \STATE - Generate a large number of evenly distributed observations on the line between the minimum and maximum values of the input $\mathbf{y}$.
    \STATE - Pass these observations through the transformation layer to obtain a sizable sample of inputs $\mathbf{y}$ and their corresponding outputs $\mathbf{\tilde{z}}$.
    \STATE - Reverse the roles, treating outputs $\mathbf{\tilde{z}}$ as new inputs and inputs $\mathbf{y}$ as outputs.
    \STATE - Train a B-spline with numerous knots and without penalization using the reversed dataset by solving the linear regression problem using ordinary least squares (OLS).
\end{algorithmic}
\end{algorithm}

\subsection{Optimization} \label{app:optimisation}

We optimize the model using a second-order optimizer, similar to the approach adopted for the MCTM \cite{MCTM}. Specifically, we employ the Limited-memory Broyden–Fletcher–Goldfarb–Shanno algorithm (LBFGS) with a Wolfe line search for the learning rate, as implemented in PyTorch \cite{pytorch}.
A crucial aspect that enhances model results and reduces training time is the pretraining of the transformation layer prior to training the complete model with decorrelation layers. This two-step approach of individually pretraining the CTM marginals before training the joint model, also used by \cite{MCTM} for the MCTM, is instrumental in finding optimal solutions for flexible model specifications and complex datasets.
For hyperparameter tuning, we utilize the Tree-structured Parzen Estimator (TPE) Sampler \cite{tpe_sampler}, implemented in the Optuna package \cite{optuna}. This sampler begins with a random search and subsequently performs targeted searches by fitting a mixture of Gaussians to previous trials in order to identify promising hyperparameters. This method is particularly suitable for our use case, as the TPE sampler, in its multivariate configuration, accounts for hyperparameter dependencies, which is significant due to the intuitive interdependence of the three penalties. 
Finally, we implement a validation set-based early stopper, which is crucial in applications aimed at maximizing the model's generalization capability.

\section{Vine Copulas} \label{app:vine_copula}

Vine copulas, or \textit{pair-copula constructions} (PCCs), decompose a multivariate copula into a series of bivariate copulas organized in a graphical structure known as a \textit{vine}. This approach facilitates modeling intricate dependencies by breaking down the multivariate problem into simpler, manageable components. There are three main types of vines: \textit{C-vines} (Canonical vines), \textit{D-vines} (Drawable vines), and \textit{R-vines} (Regular vines).
\textit{C-vines} are distinguished by a central node connected to all other nodes at each level of the decomposition, allowing for a hierarchical modeling approach where dependencies are sequentially introduced. In contrast, \textit{D-vines} present a sequential structure where dependencies are introduced in a chain-like manner, making them particularly suitable for time-series data and applications with natural ordering.
Finally, \textit{R-vines} provide a generalization that incorporates aspects of both \textit{C-vines} and \textit{D-vines}. The versatility of vine copulas lies in their ability to integrate various types of bivariate copulas at each step of the decomposition, offering a flexible framework for capturing different types of nonlinear dependencies, including tail dependencies and asymmetries.

\subsection{Proof of Theorem 1: Full Conditional Independence in R-Vines} \label{app:prof_theorem_rvine_ci}

In the following we prove Theorem \ref{theo:rvine_ci} for constructing R-Vines with full conditional independence structures.
To do so we first propose and prove Propositions \ref{prop1} and \ref{prop2}.

\begin{proposition} \label{prop1}
 Consider a $J$-dimensional random vector $\mathbf{Y}$ with joint density $f_{1,\ldots,J}(y_1,\ldots,y_J)$. If the joint distribution is represented as a truncated R-vine $(\mathcal{F}, \mathcal{V}, \mathcal{B})$ with marginal distributions $\mathcal{F}$, the regular vine tree sequence $\mathcal{V}$ and the bivariate pair copulas $\mathcal{B}$, a pair $(Y_u,Y_v)$ of elements of $\mathbf{Y}$ is fully conditionally independent if any pair copula $c_{l,m|\mathbf{J}}\in\mathcal{B}$ that involves both indices $u$ and $v$ in either, $l$, $m$ or the conditioning set $\mathbf{J}$ is an independence copula.
\end{proposition}
\begin{proof}
    As described in Theorem 5.15 of \cite{czado_vine_book}, the joint density $f_{1,\ldots,J}(y_1,\ldots,y_J)$ of a truncated R-vine structure can be represented as
    \begin{equation}\label{rvinedensity}
    f_{1,\ldots,J}(y_1,\ldots,y_J) = f_1(y_1) \cdot\ldots\cdot f_J(y_J)
    \prod_{j=1}^{J-1} \prod_{e \in E_j} c_{\mathcal{C}_{e,a}, \mathcal{C}_{e,b};\mathbf{J}_e}(F_{\mathcal{C}_{e,a} | \mathbf{J}_e}(y_{e,a}|y_{\mathbf{J}_e}), F_{\mathcal{C}_{e,b} | \mathbf{J}_e}(y_{e,b}|y_{\mathbf{J}_e}) )
    \end{equation}
    i.e.\ as the product of the densities of the marginal distributions and the bivariate copula densities across all trees, where $E_j$ is the edge set of the Tree $T_j$ in the R-vine tree sequence $\mathcal{V}$. Furthermore, $c_e = c_{\mathcal{C}_{e,a} \mathcal{C}_{e,b};\mathbf{J}_e}$ is the pair copula density corresponding to edge $e$ connecting nodes $y_a$ and $y_b$ conditioned on nodes $y_{\mathbf{J}_e}$. The corresponding copula density $c_e$ has the cumulative distribution function values $F_{\mathcal{C}_{e,a} | \mathbf{J}_e}$ and $F_{\mathcal{C}_{e,b} | \mathbf{J}_e}$ as arguments, hence the subscription $c_e = c_{\mathcal{C}_{e,a} \mathcal{C}_{e,b};J_e}$ to highlight that computing the density $c_e$ requires the conditioning sets $\mathcal{C}_{e,a} = y_a | y_{\mathbf{J}_e}$ and $\mathcal{C}_{e,b} = y_b | y_{\mathbf{J}_e}$.

    We then have the following two preliminary results:
    \begin{itemize}
    \item[(P1)]         
    For the independence copula, the copula density is given by
    $$
    c^{indep}_{\mathcal{C}_{e,a}, \mathcal{C}_{e,b};\mathbf{J}_e}(F_{\mathcal{C}_{e,a} | \mathbf{J}_e}(y_{e,a}|y_{\mathbf{J}_e}), F_{\mathcal{C}_{e,b} | \mathbf{J}_e}(y_{e,b}|y_{\mathbf{J}_e}) ) = F_{\mathcal{C}_{e,a} | \mathbf{J}_e}(y_{e,a}|y_{\mathbf{J}_e}) \cdot F_{\mathcal{C}_{e,b} | \mathbf{J}_e}(y_{e,b}|y_{\mathbf{J}_e}),
    $$
    i.e. the joint density reduces to the product of the cumulative distribution functions. 
    \item[(P2)] The conditional cumulative distribution function $F_{\mathcal{C}_{e,a} | \mathbf{J}_e}(y_{e,a}|y_{J_e})$ can be related to the pair copula $c_r$, $r \in E_{j-1}$ from the previous level of the tree sequence $E_{j-1}$. Here, $c_r$ connects nodes $y_a$ and $y_j$ conditioned on $\mathbf{J}_r = \mathbf{J}_e / \{v\}$:
    $$
    F_{\mathcal{C}_{e,a} | \mathbf{J}_e}(y_{e,a}|y_{\mathbf{J}_e}) = \frac{\partial}{\partial F_{\mathcal{C}_{r,j} | \mathbf{J}_r}(y_{r,j}|y_{\mathbf{J}_r})} C_{\mathcal{C}_{r,a}, \mathcal{C}_{r,j};\mathbf{J}_r}(F_{\mathcal{C}_{r,a} | \mathbf{J}_r}(y_{r,a}|y_{\mathbf{J}_r}), F_{\mathcal{C}_{r,j} | \mathbf{J}_r}(y_{r,j}|y_{\mathbf{J}_r}) )
    $$
    If $\mathcal{C}_{r,j}$ is an independence copula, then this simplifies to
    $$
    F_{\mathcal{C}_{e,a} | \mathbf{J}_e}(y_{e,a}|y_{\mathbf{J}_e}) = F_{\mathcal{C}_{r,a} | \mathbf{J}_r}(y_{r,a}|y_{\mathbf{J}_r}).
    $$
    \end{itemize}

    Based on these preliminaries, we now go through all copulas $c_e$ in density (\ref{rvinedensity}) and check if they factorize with respect to $y_u$ and $y_j$ since then the complete density also factorizes. 
    We have the following cases:
    \begin{enumerate}
        \item[(i)] If $l,m,\mathbf{J}$ comprises none or only one of the indices $u,v$, then obviously the copula density can be factorized since it involves at most one of the indices $u,v$ and is constant with respect to the other.
        \item[(ii)] If the indices $l,m$ are given by $u,v$, then neither $u$ nor $v$ can be in the conditioning set $\mathbf{J}$. Then, by assumption, $c_e$ is an independence copula and therefore, as shown in (P1), the copula density factorizes.
        \item[(iii)] The indices $l,m$ contain one of $u,v$ while the other index is part of the conditioning set $\mathbf{J}$. Then, by assumption, $c_e$ is an independence copula and equal to the product of conditional distributions. In addition, we have to consider the conditional distributions of the arguments of the copula density. One does not comprise any of $u,v$ (neither as the argument nor in the conditioning set) while the other, due to (P2) is the result of a copula containing $u$ and $v$ which is an independence copula and thus factorizes by being the product of conditional distributions which are again products of conditional distributions due to independence copulas until $u$ and $v$ factorize.
        \item[(iv)] If both indices $u,v$ are part of the conditioning set $\mathbf{J}$, we again use (P2) and similar arguments as in case (iii).
    \end{enumerate}
\end{proof}

\begin{proposition} \label{prop2}
If a pair of variables $Y_u,Y_v$ in an arbitrary R-vine has its pair copula $c_{u,v|\mathbf{J}}$ in Tree $T_u$, then in all pair copulas $c_{l,m|\tilde{\mathbf{J}}}$ of trees $[T_1,T_2,\ldots,T_{u-1}]$ lower in the tree sequence, $Y_u$ and $Y_v$ will never simultaneously be in the arguments $l,m$ nor the conditioning set $\tilde{\mathbf{J}}$ at the same time.
\end{proposition}
\begin{proof}
    Based on the detailed formulas derived for proving Proposition \ref{prop1}, we go through all possible cases:
     \begin{enumerate}
        \item[(i)] If $l,m,\tilde{\mathbf{J}}$ contain none or at most of the indices $u,v$, there is nothing to show.
        \item[(ii)] $u,v$ can not be the arguments $l,m$ of $c_{l,m|\tilde{\mathbf{J}}}$ since each pair of variables is only modeled once in a vine sequence and therefore if $c_{u,v|\mathbf{J}}$ is in tree $T_u$ it cannot be in trees $[T_1,T_2,\ldots,T_{u-1}]$     
        \item[(iii)] If one index is an argument to $c_{l,m|\tilde{\mathbf{J}}}$ while the other index is in the conditioning set, we would need the copula $c_{u,v|\mathbf{J}}$ from tree $T_u$ to determine the conditional CDF $F_{\mathcal{C}_{e,l} | J_e}(y_{e,l}|y_{J_e})$ which is not possible since it would violate the conditions of a valid vine structure.
    \end{enumerate}
\end{proof}
With the two propositions \ref{prop1} and \ref{prop2}, we can now prove Theorem \ref{theo:rvine_ci}:
\begin{proof}
    Jue to Proposition \ref{prop2}, we know that $Y_u$ and $Y_v$ are not jointly modeled in Trees $[T_1,T_2,\ldots,T_{u-1}]$ for an arbitrary R-vine structure.
    Hence, under the conditional independence definition of Proposition~\ref{prop1} they are conditionally independent based on this truncated model.
    Furthermore, as by assumption all pair copulas in trees $[T_i,T_{u+1},\ldots,T_{J-1}]$ are independence copulas, all these copulas factorise into conditional marginals that only depend on the conditioning variables of trees $[T_1,T_2,\ldots,T_{u-1}]$ to which $Y_u$ and $Y_v$ do not belong. As a consequence, no pair copula in Trees $[T_1,T_2,\ldots,T_{u-1}]$ contains both $Y_u$ and $Y_v$ and no marginal in $[T_i,T_{u+1},\ldots,T_{J-1}]$ depends on $Y_u$ and $Y_v$ at the same time. 
    Taken together, this results in all factors of the joint density never containing both $Y_u$ and $Y_v$ which makes the pair of variables $Y_u,Y_v$ fully conditionally independent, i.e. $Y_u \independence Y_v | \mathbf{Y}_{-u,v}$.
\end{proof}

\section{Appendix: Simulation Study} \label{app:vine_sim}

\begin{figure}[ht]
    \centering
    \footnotesize{$\kld$} \\
    \begin{minipage}{0.32\textwidth}
        \centering
        \includegraphics[width=\textwidth]{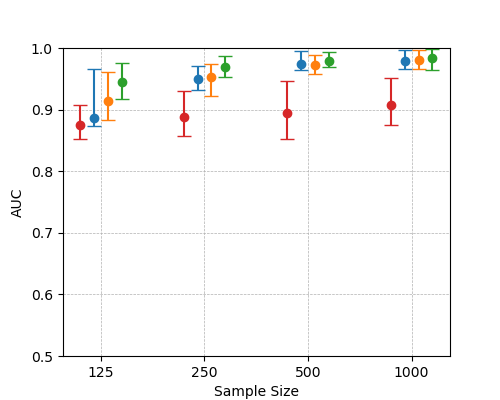}
        \subcaption{D-vine}
    \end{minipage}
    \begin{minipage}{0.32\textwidth}
        \centering
        \includegraphics[width=\textwidth]{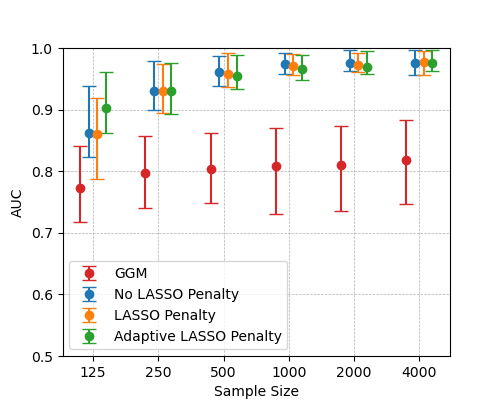}
        \subcaption{R-vine}
    \end{minipage}
    \begin{minipage}{0.32\textwidth}
        \centering
        \includegraphics[width=\textwidth]{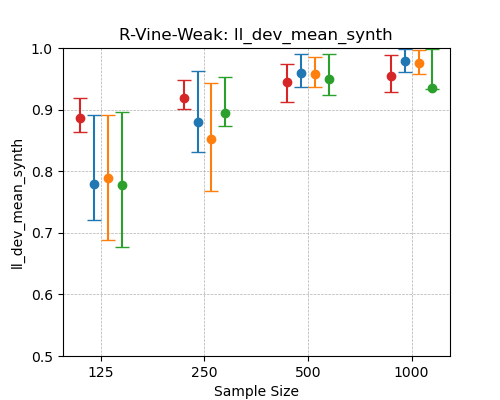}
        \subcaption{R-vine-Weak}
    \end{minipage}
    \begin{minipage}{0.32\textwidth}
        \centering
        \includegraphics[width=\textwidth]{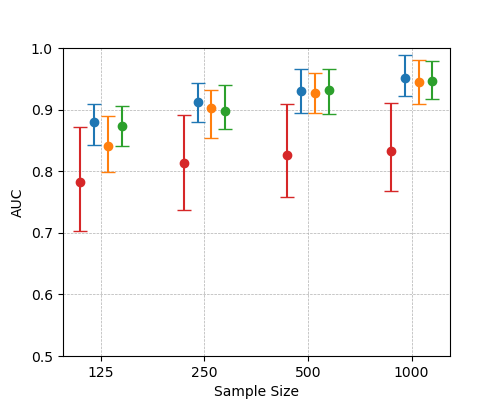}
        \subcaption{R-vine-Positive}
    \end{minipage}
    \begin{minipage}{0.32\textwidth}
        \centering
        \includegraphics[width=\textwidth]{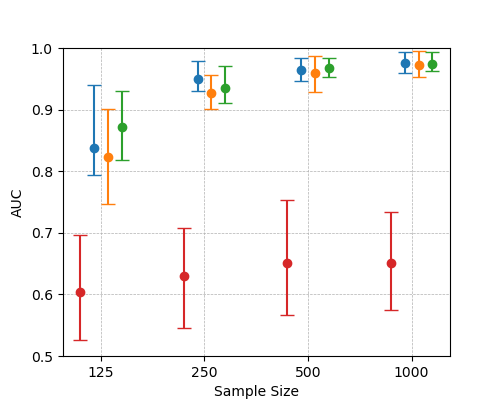}
        \subcaption{C-vine}
    \end{minipage}
    \caption{Across different vine Scenarios, the figure depicts the $\auc$ of the GTM in identifying full conditional independencies based on the $\kld$ between the GTM and the GTM with independence assumption for each dimension pair. The GTM performance for different training sample sizes is presented for the model without a lasso penalty in blue, with a lasso penalty in orange and with an adaptive lasso penalty in green, and the benchmark GGM in red. The $\kld$ scores are approximated via $10.000$ synthetic samples from the GTM. The dots represent the means and the whiskers the $20\%$ and $80\%$ quantiles across $30$ seeds.}
    \label{app_fig:xvine_auc_kld}
\end{figure}

\begin{figure}[ht]
    \centering
    \footnotesize{$\rho_{r,c,-}$} \\
    \begin{minipage}{0.32\textwidth}
        \centering
        \includegraphics[width=\textwidth, scale=1]{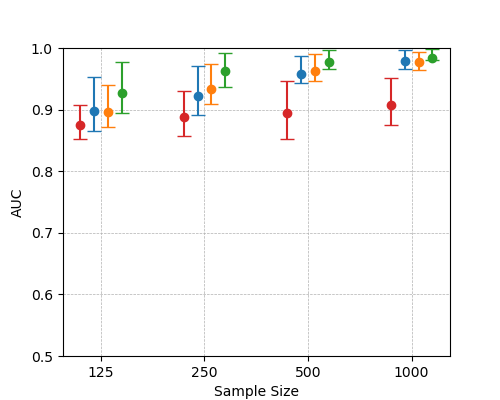}
        \subcaption{D-vine}
    \end{minipage}
    \begin{minipage}{0.32\textwidth}
        \centering
        \includegraphics[width=\textwidth, scale=1]{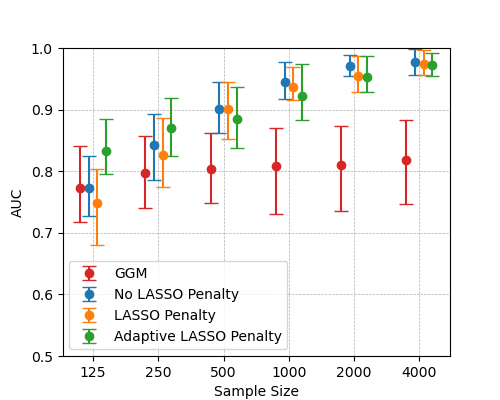}
        \subcaption{R-vine}
    \end{minipage}
    \begin{minipage}{0.32\textwidth}
        \centering
        \includegraphics[width=\textwidth, scale=1]{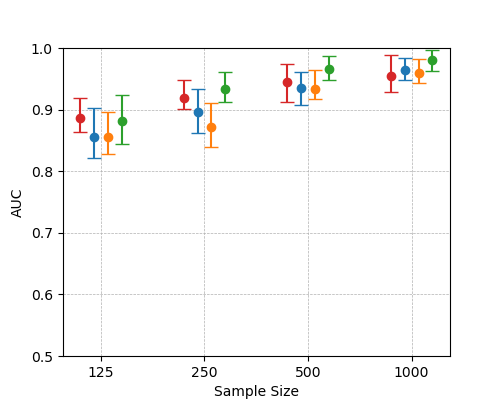}
        \subcaption{R-vine-Weak}
    \end{minipage}
    \begin{minipage}{0.32\textwidth}
        \centering
        \includegraphics[width=\textwidth, scale=1]{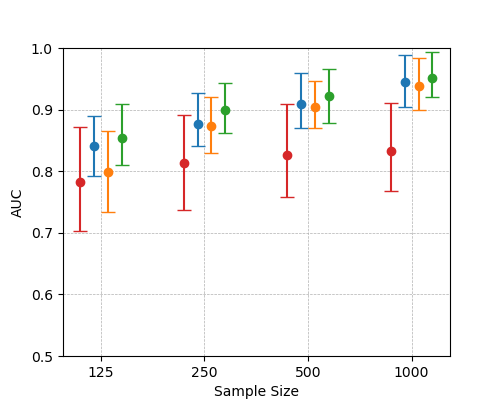}
        \subcaption{R-vine-Positive}
    \end{minipage}
    \begin{minipage}{0.32\textwidth}
        \centering
        \includegraphics[width=\textwidth, scale=1]{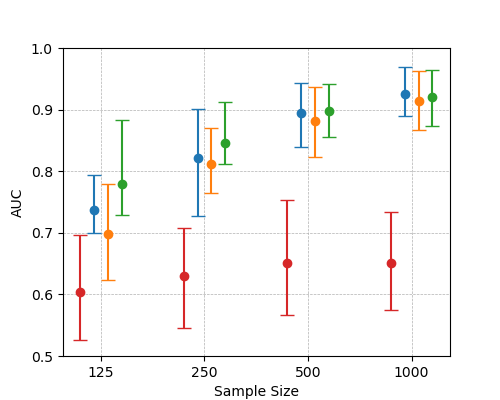}
        \subcaption{C-vine}
    \end{minipage}
    \caption{Across different vine Scenarios, the figure depicts the $\auc$ of the GTM in identifying full conditional independencies based on the absolute mean local pseudo conditional correlation $\rho_{r,c,-}$ between the GTM and the GTM with independence assumption for each dimension pair. The GTM performance for different training sample sizes is presented for the model without a lasso penalty in blue, with a lasso penalty in orange and with an adaptive lasso penalty in green and the benchmark GGM in red. The absolute mean local pseudo conditional correlation scores are approximated via $10.000$ synthetic samples from the GTM. The dots represent the means and the whiskers the $20\%$ and $80\%$ quantiles across $30$ seeds.}
    \label{app_fig:xvine_auc_condcorr}
\end{figure}

\begin{figure}[ht]
    \centering
    \footnotesize{$p_{r,c,-}$} \\
    \begin{minipage}{0.32\textwidth}
        \centering
        \includegraphics[width=\textwidth]{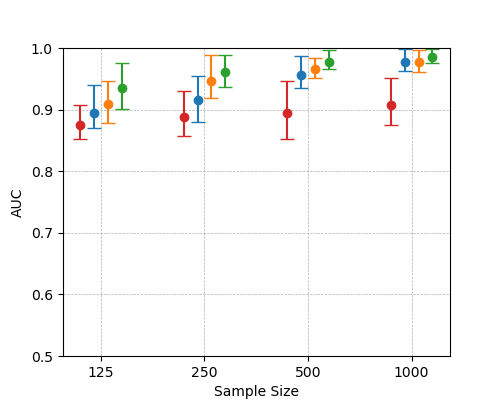}
        \subcaption{D-vine}
    \end{minipage}
    \begin{minipage}{0.32\textwidth}
        \centering
        \includegraphics[width=\textwidth]{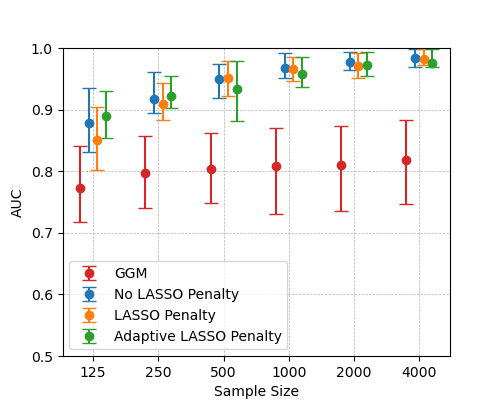}
        \subcaption{R-vine}
    \end{minipage}
    \begin{minipage}{0.32\textwidth}
        \centering
        \includegraphics[width=\textwidth]{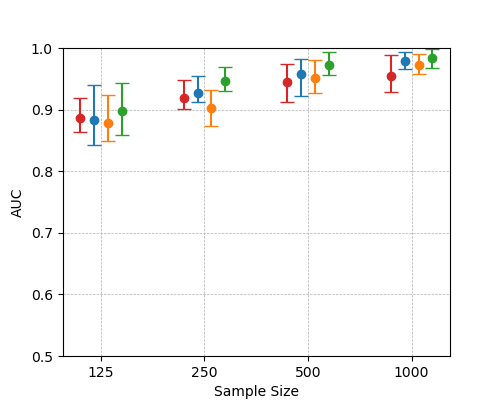}
        \subcaption{R-vine-Weak}
    \end{minipage}
    \begin{minipage}{0.32\textwidth}
        \centering
        \includegraphics[width=\textwidth]{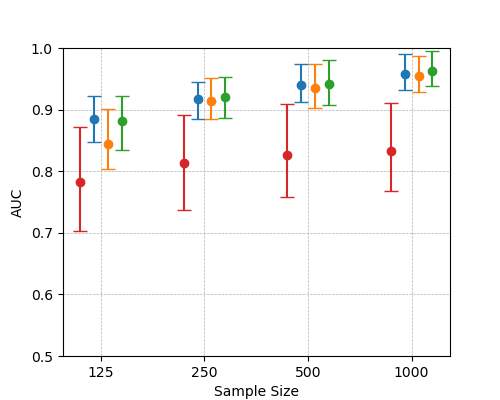}
        \subcaption{R-vine-Positive}
    \end{minipage}
    \begin{minipage}{0.32\textwidth}
        \centering
        \includegraphics[width=\textwidth]{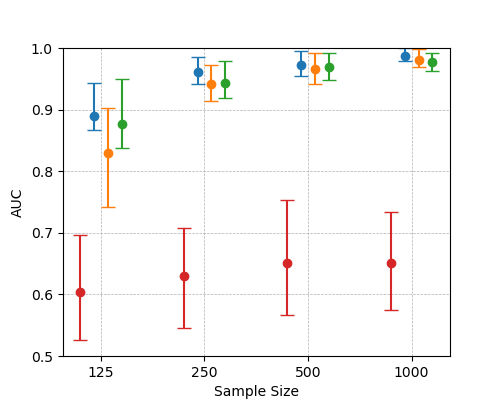}
        \subcaption{C-vine}
    \end{minipage}
    \caption{Across different vine Scenarios, the figure depicts the $\auc$ of the GTM in identifying full conditional independencies based on the absolute mean local pseudo precision matrix $p_{r,c,-}$ entry between the GTM and the GTM with independence assumption for each dimension pair. The GTM performance for different training sample sizes is presented for the model without a lasso penalty in blue, with a lasso penalty in orange and with an adaptive lasso penalty in green, and the benchmark GGM in red. The absolute mean local pseudo precision matrix scores are approximated via $10.000$ synthetic samples from the GTM. The dots represent the means and the whiskers the $20\%$ and $80\%$ quantiles across $30$ seeds.}
    \label{app_fig:xvine_auc_pmatrix}
\end{figure}

\FloatBarrier

\section{Appendix: MAGIC Application} \label{app:magic}

In the following \autoref{alg:cond_samples} we describe how to conditionally sample from the GTM:
\begin{algorithm}[H]
\caption{Sampling two dimensional conditional values}
\label{alg:cond_samples}
\textbf{Input:} Sampling point vector $\mathbf{y}^c \in \mathbb{R}^{J}$, dimensions to sample $u, v$, trained GTM $\mathbf{h}(\bfy)$ with the probability density $f(\bfy)$.\\
\textbf{Output:} conditional samples $(y_u^*, y_v^*)$
\begin{algorithmic}
\STATE - Fix all dimensions except $u$ and $v$: $\bfy_{\backslash \{u,v\}}^c$ remains unchanged.
\STATE - Sample $S$ new candidates $(y_u^{(s)}, y_v^{(s)})$ from their marginal distributions, this can be done by simply sampling from the GTM:
\[
    y_u^{(s)} \sim f(y_u), \quad y_v^{(s)} \sim f(y_v) \quad \text{for } s = 1, \dots, S
\]
\STATE - Compute the marginal density of the conditioning set:
\[
f(\bfy_{\backslash \{u, v\}}^c) \stackrel{\text{GLQ}}{\approx}
\iint f(\bfy^c) \, dy_u \, dy_v
\]
\FORALL {proposed conditional samples $y_u^s, y_v^s$ with $s \in [1,...,S]$}
    \STATE - Compute the conditional density:
    \[
        f(y_u^s, y_v^s | \bfy_{\backslash \{u, v\}})^s = \frac{
        f(y_u^s, y_v^s, \mathbf{y}_{- \{u,v\}})
        }{
        f(\bfy_{\backslash \{u, v\}})
        }
    \]
\ENDFOR
\STATE - Accept samples $(y_u^*, y_v^*)$ from the candidate set $(y_u^s, y_v^s)$ according to there probability $f(y_u^s, y_v^s | \bfy_{\backslash \{u, v\}})^s$.
\end{algorithmic}
\end{algorithm}
One dimensional conditional samples can be sampled accordingly. 
For higher dimensions, where higher dimensional GLQ are unfeasible, one can resort to importance sampling.
In doing so one uses the density values of the proposed samples $f(y_u^s, y_v^s, \mathbf{y}_{- \{u,v\}})$ as the weight $w_s$, then normalizes the weights and samples $(y_u^*, y_v^*)$ from the $(y_u^s, y_v^s)$ according to the weights $w_s$.

\begin{figure}[htbp]
  \centering
  \begin{subfigure}[b]{0.45\linewidth}
  \caption{$Y$}
    \includegraphics[width=\linewidth]{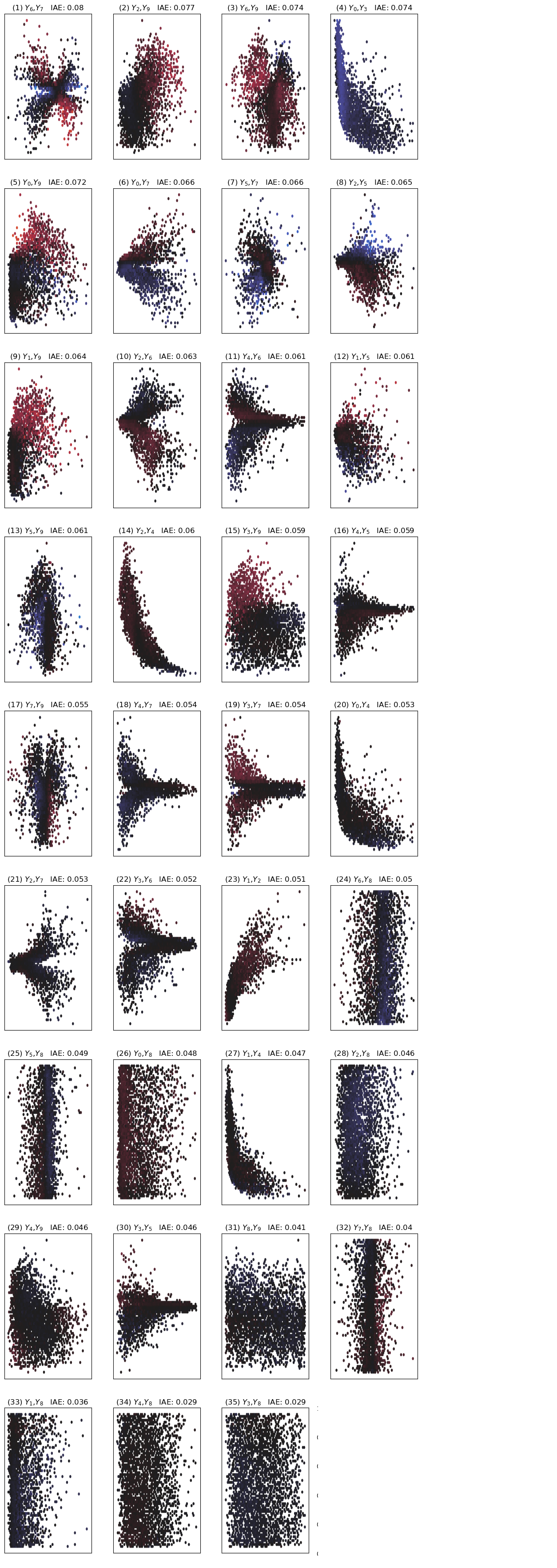}
  \end{subfigure}
  \begin{subfigure}[b]{0.45\linewidth}
  \caption{$\tilde{Z}$}
    \includegraphics[width=\linewidth]{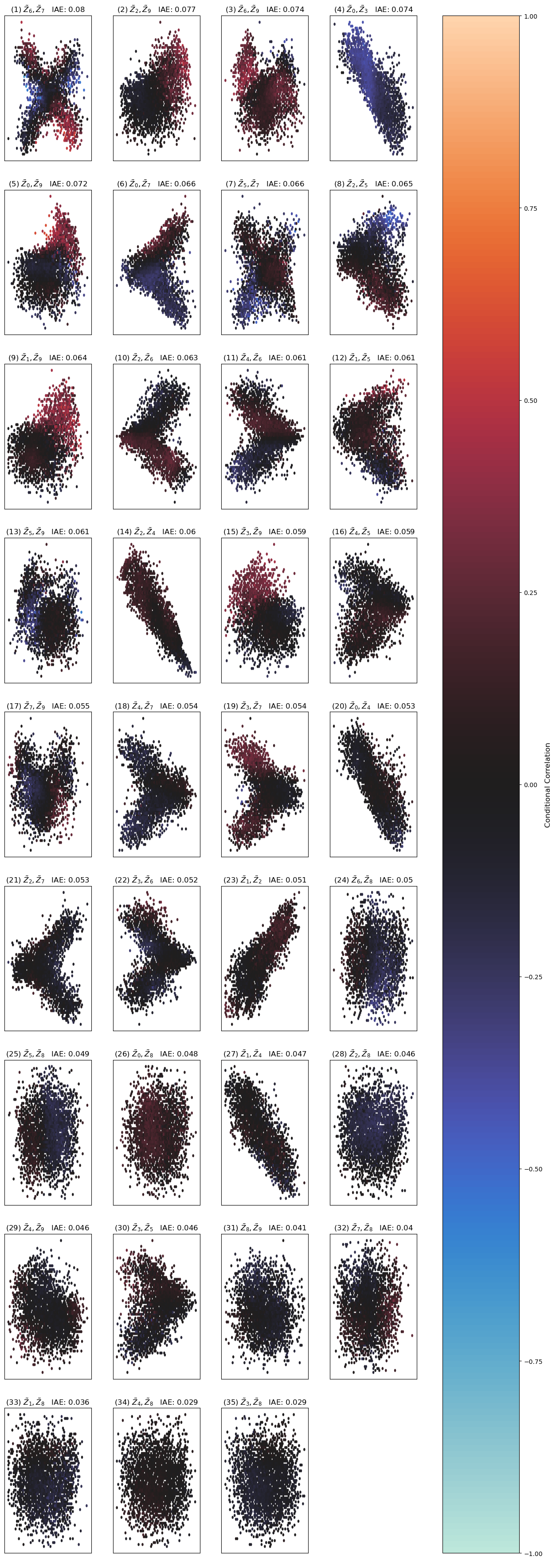}
  \end{subfigure}
  \caption{The 37 conditionally independent pairs, measured by an $\iae < 0.1$, in the hadron dataset ordered by their $\iae$ conditional dependence metrics.
  The plots on the left depict the original training data and the ones on the right depict the data after the marginal transformation.
  Both are colored by the local conditional pseudo-correlations.}
  \label{fig:unsign_conditional_corr_h}
\end{figure}

\FloatBarrier

\end{document}